\documentclass[11pt,a4paper]{article}

\usepackage[a4paper,margin=2.5cm]{geometry}

\usepackage{amsmath,amssymb,amsfonts,amsthm,mathtools}

\usepackage[T1]{fontenc}
\usepackage{lmodern}                % clean Computer‑Modern in T1
\usepackage{microtype}              % better typesetting
\usepackage[colorlinks,citecolor=blue,linkcolor=blue,urlcolor=blue]{hyperref}
\usepackage[nameinlink]{cleveref}   % \cref, \Cref auto-formats
\usepackage{lipsum}

\usepackage{booktabs}
\usepackage{array}
\usepackage{float}
\usepackage{graphicx}
\usepackage{subcaption}

\usepackage{algorithm}
\usepackage{algorithmicx}
\usepackage{algpseudocode}

\usepackage{enumitem}

\usepackage{natbib}

\newtheorem{theorem}{Theorem}[section]
\newtheorem{proposition}[theorem]{Proposition}
\newtheorem{lemma}[theorem]{Lemma}
\newtheorem{corollary}[theorem]{Corollary}
\theoremstyle{definition}
\newtheorem{definition}[theorem]{Definition}

\theoremstyle{remark}
\newtheorem{remark}[theorem]{Remark}

\newcommand{\MF}{\mathfrak{M}_F}
\newcommand{\MFcross}{\mathfrak{M}_F^{(\mathrm{cross})}}
\newcommand{\R}{\mathbb{R}}

\newcommand{\op}{\mathrm{op}}

\newcommand{\bX}{\mathbf{X}}
\newcommand{\bZ}{\mathbf{Z}}

\newcommand{\dmin}{d_{\min}}
\DeclareMathOperator{\vecop}{vec}

\title{Information-Theoretic Bounds for Sparse Covariance Estimation\\
       in the Vertical-Split Distributed Model}

\author{%
  Jing Yee Tan\thanks{Department of Mathematics, University of Hong Kong.
  Email: \texttt{jyt88@connect.hku.hk}.}
  \and
  Guangyue Han\thanks{Department of Mathematics, University of Hong Kong. Email: \texttt{ghan@hku.hk}.}
}

\date{\today}

\begin{document}
\maketitle

% =============================================================
%  ABSTRACT
% =============================================================
\begin{abstract}
We study the minimax estimation error for distributed covariance matrix estimation in the vertical-split (feature-split) setting, where two agents each observe different coordinates of~$m$ i.i.d.\ sub-Gaussian samples and communicate a limited number of bits to a central server. While \cite{rahmani2025fundamental} established nearly tight bounds for dense (unstructured) cross-covariance matrices, we investigate whether imposing elementwise $s$-sparsity on the cross-covariance $C_{21}$ can reduce the required communication and sample complexity. We show that sparsity can reduce the communication required for cross-covariance estimation in the vertical split at a fixed target distortion.

Specifically, for sufficiently large $d_1d_2/s'$ and $0<\varepsilon<\sigma^2\sqrt{s'}/32$, any scheme achieving expected Frobenius distortion at most $\varepsilon$ must satisfy $B_k = \Omega(\sigma^4 d_k\, s' \log(d_1 d_2/s')/\varepsilon^2)$ and $m = \Omega(\sigma^4\, s' \log(d_1 d_2/s')/\varepsilon^2)$ for cross-covariance estimation, where $s' = s \wedge d_{\min}$. For the $1$-sparse case, our achievable scheme reduces the $d_1d_2$ factor in the dense communication rate to $\log(d_1d_2)$, up to polylogarithmic factors, for the cross-covariance communication component in the matching regime. Our lower bounds are established via Fano's method with an explicit sparse packing using a Varshamov--Gilbert-type argument for signed partial permutation matrices combined with the Conditional Strong Data Processing Inequality of \cite{rahmani2025fundamental}. We show that the communication lower bound is tight up to polylogarithmic factors under the conditions of Remark~\ref{rem:achievmatch}, using an achievable scheme based on covering-net quantization and entry-wise hard thresholding.
\end{abstract}

\section{Introduction}
Statistical estimation in distributed environments has become a fundamental problem in modern data science.  In a typical formulation, data is spread across $K$ agents who observe independent or correlated samples of an unknown parameter; each agent compresses its local data into a limited number of bits and sends the result to a central server, which must produce an accurate estimate.  A central goal of the theory is to characterize the minimax tradeoff between the communication budget, the number of samples, and the estimation error.  For dense (unstructured) parameters, a line of work beginning with \cite{zhang2014information} and \cite{garg2014communication} established that in the horizontal split, where each of $K$ machines holds a subset of $m$ i.i.d.\ samples and observes all $d$ coordinates, the communication cost must scale linearly in the ambient dimension $d$ to achieve the statistical minimax rate.  Since high-dimensional models often possess intrinsic low-dimensional structure, these results motivate a natural meta-question: \emph{How does structural knowledge of the parameter, such as sparsity, affect the communication cost of distributed estimation?}

For distributed Gaussian mean estimation in the horizontal split,
\cite[Section~3.2, Theorem~3.2]{garg2014communication}
proposed a thresholding-based protocol exploiting the
$s$-sparsity of the mean vector $\theta\in\mathbb{R}^d$.
Up to logarithmic factors, their protocol can reduce communication
by a factor of $d/s$ while preserving the dense-model error rate,
or improve the error by that factor while retaining the
dense-model communication budget. They conjectured that this
communication--error tradeoff is optimal up to logarithmic factors.
\cite[Theorem~4.5]{braverman2016communication} established a
matching lower bound for $d\ge2s$, even in the fully interactive
(multi-round) communication model.

These results distinguish communication savings at a fixed
accuracy from the communication required to attain the improved
sparse statistical rate. Although sparsity can reduce
communication at a fixed target error, attaining the sparse
statistical minimax rate can still require communication with
ambient-dimension dependence, up to logarithmic factors
\cite[Theorem~4.5]{braverman2016communication}.
Thus, the remaining limitation concerns the dependence on the
ambient dimension, rather than an absence of communication
savings from sparsity. This motivates a broader question:
\emph{how does sparsity affect the communication--accuracy
tradeoff in other distributed estimation problems?}

The \emph{vertical} (feature) split provides a fundamentally different distributed model. Here, all $K$ agents observe the \emph{same} set of $m$ samples, but each agent sees only a subset of coordinates: Agent~$i$ observes the features $\{X_i^{(j)}\}_{j=1}^m$ corresponding to dimensions $d_i$, with $\sum_i d_i=d$. The natural estimation target is the \emph{cross-covariance matrix} $C_{21}=\mathbb{E}[(X_2-\mathbb{E}X_2) (X_1-\mathbb{E}X_1)^\top]\in\mathbb{R}^{d_2\times d_1}$, since the diagonal blocks $C_{kk}$ can be estimated locally without communication. \cite{rahmani2025fundamental} recently characterized the minimax distortion for dense (unstructured) covariance estimation in this setting with $K=2$ agents: For $0<\varepsilon<\sigma^2\sqrt{d_{\min}}/224$, under the Frobenius norm, the communication budget per agent must scale as $B_k=\Omega(\sigma^4 d_1 d_2d_k/\varepsilon^2)$, and they constructed an achievable scheme matching this bound up to polylogarithmic factors.  A key feature of the vertical split is that the agents' data is \emph{correlated within each sample}: $(X_1^{(t)}, X_2^{(t)})$ are jointly distributed for every $t$. Braverman's D-SDPI, which requires conditional independence across machines, therefore does not apply.  To handle this correlation, Rahmani et al.\ introduced the \emph{Conditional Strong Data Processing Inequality} (C-SDPI), which quantifies information contraction through a state-dependent channel.

The \emph{vertical} (feature) split is a natural model for settings in
which multiple parties hold different attributes of the same individuals. In healthcare, one institution may record genomic data while another stores clinical outcomes for the same patient cohort; in finance, a bank observes
transaction histories while a credit bureau holds repayment records for the same customers (\cite{rahmani2025fundamental}).  In such settings, the
quantity of primary interest is the \emph{cross-covariance} between the
feature sets held by different agents, for instance, the correlation
between genetic markers and treatment response, or between spending patterns and default risk, since each agent can estimate its own marginal statistics locally without communication.  The vertical split also arises in multi-sensor systems where different instruments measure different physical quantities of the same phenomenon.  Unlike the horizontal split, where the communication channel transmits information about independent random variables, the vertical-split channel is \emph{state-dependent}: the message sent by one agent is statistically dependent on the data held by the other, because both observe the same underlying samples. This correlation structure makes the vertical split both practically important and information-theoretically distinct from the well-studied horizontal setting.

A natural question left open by~\cite{rahmani2025fundamental} is whether \emph{sparsity} in the cross-covariance $C_{21}$ can reduce the communication cost in the vertical split.  If $C_{21}$ has at most $s$ non-zero entries (equivalently, at most $s$ out of $d_1 d_2$ cross-agent feature pairs are correlated), does the communication budget drop?

In this paper, we show that \emph{sparsity can reduce the communication required for cross-covariance estimation in the vertical split at a fixed target distortion}.  For $s$-sparse cross-covariance estimation with $s\le d_{\min}$, in the common validity regime of the dense lower bound and our matching remark, our scheme achieves a cross-covariance communication component 
% the communication budget per agent drops from $\Omega(\sigma^4 d_1 d_2 d_k/\varepsilon^2)$ in the dense case to 
$O(\beta\sigma^4d_ks\log(d_1d_2)/\varepsilon^2)$, compared with the dense communication lower bound $B_k=\Omega(\sigma^4d_kd_1d_2/\varepsilon^2)$, an improvement from $d_1 d_2$ to $s\log(d_1 d_2/s)$ up to logarithmic factors. The sparsity dependence in our lower bound comes from the sparse packing construction. Restricting to $s'$-sparse hypotheses changes the logarithmic packing-cardinality lower bound from order $d_1 d_2$ to order $s'\log(d_1 d_2/s')$, and this change propagates directly into the communication lower bound.

This result has practical implications for distributed systems in which cross-agent dependencies are inherently sparse. A similar sparsity arises in high-dimensional multi-omics studies, where one measures hundreds or thousands of features (e.g., metabolites or proteins) on a handful of samples. There, the correlation structure among the measured features is sparse and block-structured: only variables belonging to a common biological process are strongly associated, while the vast majority of pairs are effectively independent (\cite{perrot2019estimation}).

\subsection{Main Contributions}
We study the minimax estimation error for distributed covariance matrix estimation (DCME) in the vertical-split setting under elementwise sparsity constraints on the cross-covariance matrix $C_{21}$.  Let $s':= s\wedge(d_1\wedge d_2)$ and for large enough $d_1d_2/s'$, write $L=\Theta(\log(d_1 d_2/s'))$.  Our main contributions are:
\begin{enumerate}[label=(\roman*)]
\item \textbf{$s$-Sparse minimax lower bounds}
(Theorems~\ref{thm:ssparse-dccme} and~\ref{thm:ssparse-dcme}).
For cross-covariance estimation under the Frobenius norm,
\[
    \mathfrak{M}_F^{(\mathrm{cross})}
    =\Omega\!\left(\sigma^2\left[
    \sqrt{s'\log\!\Big(\frac{d_1 d_2}{s'}\Big)
    \Big(\frac{d_1}{B_1}\vee\frac{d_2}{B_2}\Big)}\land \sqrt{s'}\right]\right).
\]
In particular, in this regime, any scheme achieving expected distortion at most $\varepsilon$, where $0<\varepsilon<\sigma^2\sqrt{s'}/32$, must satisfy $B_k=\Omega(\sigma^4 d_k\,s'\log(d_1 d_2/s')/\varepsilon^2)$, and the sample complexity for cross-covariance estimation must satisfy $m=\Omega(\sigma^4 s'\log(d_1 d_2/s')/\varepsilon^2)$. The proof machinery connects the Varshamov--Gilbert type argument in \cite{raskutti2011minimax} to the Fano method.
 
\item \textbf{Matching achievable scheme} (Theorem~\ref{thm:sachievability}).
We construct a DCME protocol that combines the covering-net quantisation framework of~\cite{rahmani2025fundamental} with entry-wise hard thresholding inspired by~\cite{garg2014communication}.  The protocol matches the communication lower bound up to polylogarithmic factors under the conditions of Remark~\ref{rem:achievmatch}. We utilize a direct entry-wise Frobenius analysis combined with hard thresholding, requiring concentration arguments for sub-exponential products that are not needed in the dense case.
\item \textbf{Sparse packing for partial permutation matrices}
(Lemma~\ref{lem:VGpartialperm}). We prove a combinatorial packing result for $s'$-sparse signed partial permutation matrices in $\{-1,0,+1\}^{d_2\times d_1}$: there exist at least $((2d_1 d_2)^{3/4}/(4es'))^{s'}$ such matrices that are pairwise separated in Hamming distance by more than $s'/2$.  This adapts the greedy construction of \cite{raskutti2011minimax} from $\ell_0$-sparse vectors to the partial permutation constraint required by the C-SDPI framework, and is the combinatorial ingredient enabling the $s$-sparse lower bound.
 
\end{enumerate}

\subsection{Related Work}
\paragraph{Distributed mean estimation.}
Communication--accuracy tradeoffs for distributed mean estimation have been studied in a sequence of works, including ~\cite{zhang2014information}, ~\cite{garg2014communication}, ~\cite{braverman2016communication}, ~\cite{CW24}, and ~\cite{SFKM17}.

\paragraph{Information-theoretic techniques.}
Beyond SDPI-based approaches, alternative methodologies for lower bounds under information constraints include the $\chi^2$-contraction framework of \cite{ACT20a,ACT20b}, later unified for interactive settings by \cite{ACST23}, and geometric approaches based on Fisher information by \cite{HMOW18} and \cite{BHO20}. The SDPI itself has a long history, originating with \cite{AG76} and developed further by
\cite{Rag13,raginsky2016strong} amongst others; see~\cite{rahmani2025fundamental} for a comprehensive discussion.

\paragraph{Centralised covariance estimation under structure.}
In the non-distributed setting, minimax-optimal rates for covariance estimation under sparsity were established by ~\cite{BL08a,BL08b} (thresholding and banding), \cite{cai2010optimal} (general rates for large covariance matrices), \cite{cai2012optimal} (sparse covariance matrices), and \cite{EK08} (operator-norm consistency). Structured variants including Toeplitz (\cite{CRZ13}) assumptions have also been studied.

\paragraph{Distributed covariance and correlation estimation.}
\cite{HS19} studied distributed estimation of Gaussian
correlations in the vertical split for the scalar and vector--scalar
cases using SDPI-based lower bounds. 
\cite{ST21} studied correlation testing between a vector and a
scalar.  \cite{rahmani2025fundamental} generalised to
the full matrix setting with near-optimal bounds.  Related distributed
PCA problems have been studied in \cite{BKLW14,KVW14,BCL05}.

\subsection{Paper Organization}
The remainder of this paper is organised as follows. Section~\ref{sec:problemformulation} formulates the problem precisely. Section~\ref{sec:prelim} collects the necessary preliminaries, including the C-SDPI and Fano's inequality. Section~\ref{sec:mainresults} states all main results---lower bounds and the matching achievable scheme---and discusses their tightness. Section~\ref{sec:prooflower} proves the lower bounds. Section~\ref{sec:proofachiev} presents the achievability protocol and a proof sketch; the full proof is given in Appendix~\ref{app:achievfull}. Section~\ref{sec:discussion} concludes with a discussion and open problems. Appendices~\ref{app:verification}--\ref{app:achievfull} contain deferred
proofs and technical verifications.

% ============================================================
%  SECTIONS 2–3 — arxiv draft (v2: notation-consistent)
%
%  Convention:
%    Non-bold  X_k        = single-sample random vector (dim d_k)
%    Non-bold  X_k^{(i)}  = i-th copy
%    Bold      \bX_k      = full dataset (d_k x m matrix)
%    Non-bold  (X_2)_i    = scalar component
%    V, W                 = hypothesis/auxiliary indices
%    M_k                  = encoded message from Agent k
%
%  Requires: \newcommand{\bX}{\mathbf{X}}
%            \newcommand{\bZ}{\mathbf{Z}}
% ============================================================

\section{Problem Formulation~\label{sec:problemformulation}}
\subsection{Data Model and Vertical Split}\label{sec:datamodel}
 
Let $Z=\bigl(\begin{smallmatrix}X_1\\X_2\end{smallmatrix}\bigr) \in\mathbb{R}^d$ be a $d$-dimensional random vector with $X_1\in\mathbb{R}^{d_1}$, $X_2\in\mathbb{R}^{d_2}$, and $d=d_1+d_2$. The distribution of $Z$ belongs to the family $\mathcal{P}=\mathrm{subG}^{(d)}(\sigma)$ of $\sigma$-sub-Gaussian distributions.  We observe $m$ i.i.d.\ copies $Z^{(1)},\dots,Z^{(m)}\sim P\in\mathcal{P}$.

In the \emph{vertical} (feature) split with $K=2$ agents, Agent~1 observes $\{X_1^{(i)}\}_{i=1}^m$ (the first $d_1$ coordinates of every sample) and Agent~2 observes $\{X_2^{(i)}\}_{i=1}^m$ (the last $d_2$ coordinates).  Both agents see all $m$ samples, but different coordinates.  We write
\begin{equation}\label{eq:dataset}
    \mathbf{X}_k := \big(X_k^{(1)},\dots,X_k^{(m)}\big)
    \;\in\;\mathbb{R}^{d_k\times m}
\end{equation}
for Agent~$k$'s full dataset.

The covariance matrix of $Z$ admits the block decomposition
\begin{equation}\label{eq:covblock}
    C \;:=\; \mathbb{E}\big[(Z-\mathbb{E}Z)(Z-\mathbb{E}Z)^\top\big]
    \;=\; \begin{pmatrix}
        C_{11} & C_{12} \\ C_{21} & C_{22}
    \end{pmatrix},
\end{equation}
where $C_{kk}=\mathrm{Cov}(X_k)$ and $C_{21}=\mathbb{E}[(X_2-\mathbb{E}X_2)(X_1-\mathbb{E}X_1)^\top] \in\mathbb{R}^{d_2\times d_1}$ is the cross-covariance matrix. The diagonal blocks $C_{11}$ and $C_{22}$ can each be estimated locally by the respective agent without communication; the cross-covariance $C_{21}$ is the block that \emph{requires} inter-agent communication, making it the natural estimation target.

\subsection{Communication Model and Minimax Objective}%
\label{sec:commmodel}
 
We consider the simultaneous (one-shot, non-interactive) communication
model of~\cite{rahmani2025fundamental}.
 
\begin{definition}[DCME and DCCME Schemes]\label{def:dcme}
A \emph{Distributed Covariance Matrix Estimation} (DCME) scheme with
parameters $(\sigma,m,d_{1:2},B_{1:2})$ consists of two encoder
functions
$\mathcal{E}_k\colon\mathbb{R}^{d_k\times m}\to[1:2^{B_k}]$,
producing messages $M_k=\mathcal{E}_k(\mathbf{X}_k)$ for
$k=1,2$, and a decoder
$\mathcal{D}\colon[1:2^{B_1}]\times[1:2^{B_2}]\to\mathbb{S}_+^{d
\times d}$, outputting $\hat{C}=\mathcal{D}(M_1,M_2)$, where $\mathbb{S}_+^{d\times d}$ is the cone of symmetric positive semidefinite matrices.
A \emph{Distributed Cross-Covariance Matrix Estimation} (DCCME) scheme
replaces the decoder with
$\mathcal{D}_{21}\colon[1:2^{B_1}]\times[1:2^{B_2}]
\to\mathbb{R}^{d_2\times d_1}$, outputting $\hat{C}_{21}$.
\end{definition}
 
\noindent
The distortion is measured under the Frobenius norm:
$\mathbb{E}\bigl[\|\hat{C}-C\|_F\bigr]$ for DCME and
$\mathbb{E}\bigl[\|\hat{C}_{21}-C_{21}\|_F\bigr]$ for DCCME.
The \emph{minimax distortions} are
\begin{align}
    \MF(\sigma,m,d_{1:2},B_{1:2},s)
    &\;:=\; \inf_{\mathcal{E}_1,\mathcal{E}_2,\mathcal{D}}
    \;\sup_{P\in\mathcal{P}_s}\;
    \mathbb{E}\big[\|\hat{C}-C\|_F\big],
    \label{eq:minimaxdcme}\\
    \MF^{(\mathrm{cross})}(\sigma,m,d_{1:2},B_{1:2},s)
    &\;:=\; \inf_{\mathcal{E}_1,\mathcal{E}_2,\mathcal{D}_{21}}
    \;\sup_{P\in\mathcal{P}_s}\;
    \mathbb{E}\big[\|\hat{C}_{21}-C_{21}\|_F\big],
    \label{eq:minimaxdccme}
\end{align}
where $\mathcal{P}_s$ is the sparse distribution class defined below.
Since $\|\hat{C}-C\|_F^2\geq 2\|\hat{C}_{21}-C_{21}\|_F^2$, any
DCME scheme induces a DCCME scheme, giving
$\MF\geq\MF^{(\mathrm{cross})}$.

\subsection{Sparsity Model}\label{sec:sparsitymodel}
We impose an \emph{elementwise hard sparsity} constraint on the
cross-covariance.
\begin{definition}[Support]
For a matrix $E \in \R^{d_2 \times d_1}$, the \emph{support} is
$S(E) := \{(j,i) \in [d_2] \times [d_1] : E_{ji} \neq 0\}$,
and $|E|_0 := |S(E)|$ is the number of nonzero entries.
\end{definition}

\begin{definition}[Symmetric difference]
For two supports $S, S' \subset [d_2] \times [d_1]$, the
\emph{symmetric difference} is
$S \triangle S' := (S \setminus S') \cup (S' \setminus S)$,
i.e., positions belonging to exactly one of $S$ and $S'$.
\end{definition}

\begin{definition}[Sparse Cross-Covariance Class]\label{def:sparsity}
For an integer $1\leq s\leq d_1d_2$, the $s$-sparse cross-covariance
class is
\[
    \mathcal{S}(s)
    \;:=\;\big\{A\in\mathbb{R}^{d_2\times d_1}
    :|A|_0\leq s\big\},
\]
and the corresponding
distribution class is
\[
    \mathcal{P}_s
    \;:=\;\big\{P\in\mathrm{subG}^{(d)}(\sigma)
    :C_{21}(P)\in\mathcal{S}(s)\big\}.
\]
The dense setting of~\cite{rahmani2025fundamental} corresponds to
$s=d_1d_2$.
\end{definition}
 
% \begin{remark}[Hard vs.\ soft sparsity]\label{rem:hardvssoft}
% We emphasise that our sparsity model is the \emph{hard} (exact)
% $\ell_0$ constraint $|C_{21}|_0\leq s$, requiring at most $s$ entries
% to be non-zero.  This contrasts with the \emph{soft} sparsity models
% common in centralised covariance estimation, such as the $\ell_q$-ball
% constraint $\sum_{i,j}|C_{ij}|^q\leq R_q$ for $q\in(0,1]$ studied in~\cite{raskutti2011minimax}.  Soft sparsity permits many small but non-zero entries, whereas hard sparsity forces exact zeros.  We work with hard sparsity for two reasons.  First, for jointly Gaussian $Z$, $(C_{21})_{ij}=0$ if and only if $(X_2)_i\perp\!\!\!\perp(X_1)_j$, so each zero entry corresponds to a pair of cross-agent features that are statistically independent; an $s$-sparse cross-covariance means that at most $s$ out of $d_1 d_2$ feature pairs carry any correlation. This sparsity--independence equivalence is the structural property that interacts with the vertical split and drives the improvement in our bounds.  Second, the hypothesis families in our lower bound proofs are constructed from matrices in $\{-1,0,+1\}^{d_2\times d_1}$, and the packing analysis relies on exact zeros to ensure each hypothesis lies in~$\mathcal{P}_s$.
% \end{remark}
\begin{remark}[Hard vs.\ soft sparsity]\label{rem:hardvssoft}
We emphasise that our sparsity model is the \emph{hard} (exact)
$\ell_0$ constraint $|C_{21}|_0\leq s$, requiring at most $s$ entries
to be non-zero. This contrasts with \emph{soft} sparsity constraints,
such as the $\ell_q$-ball constraint
$\sum_{i,j}|C_{ij}|^q\leq R_q$ for $q\in(0,1]$, analogous to the
vector $\ell_q$-ball constraints studied for linear regression
in~\cite{raskutti2011minimax}. Soft sparsity permits many small but
non-zero entries, whereas hard sparsity forces exact zeros.
We work with hard sparsity for two reasons. First, for jointly
Gaussian $Z$, $(C_{21})_{ij}=0$ if and only if
$(X_2)_i\perp\!\!\!\perp(X_1)_j$, so each zero entry corresponds
to a pair of cross-agent features that are statistically independent;
an $s$-sparse cross-covariance means that at most $s$ out of
$d_1d_2$ feature pairs have non-zero covariance. For general
sub-Gaussian $Z$, zero covariance need not imply independence.
Our achievable scheme exploits exact zeros through thresholding,
rather than relying on the Gaussian sparsity--independence
equivalence. Second, the hypothesis families in our cross-covariance
lower bound proofs are constructed from matrices in
$\{-1,0,+1\}^{d_2\times d_1}$, and the packing analysis relies on
exact zeros to ensure each hypothesis satisfies the sparsity
constraint defining~$\mathcal{P}_s$.
\end{remark}

\subsection{Notation}
We write $a\vee b:=\max(a,b)$ and $a\wedge b:=\min(a,b)$.  For a positive integer $n$, $[n]=[1:n]:=\{1,\dots,n\}$.  For a matrix $A$, $\|A\|_{\mathrm{op}}$ and $\|A\|_F$ denote the operator and Frobenius norms, respectively. For symmetric matrices, $A\succeq B$ means $A-B$ is positive semidefinite, and $A\succ B$ means $A-B$ is positive definite.  Throughout, $\log$ denotes the base-2 logarithm.  We write $s':=s\wedge d_{\min}$ where $d_{\min}:=d_1\wedge d_2$. 
% Throughout, $C_{ij}$ and $C_{X_iX_j}$ are used interchangeably for $i,j=1,2$.

Non-bold letters ($Z$, $X_k$, $X_k^{(i)}$) denote random vectors or
their individual copies; bold letters
($\mathbf{Z}$, $\mathbf{X}_k$) denote the full $m$-sample dataset as in~\eqref{eq:dataset}.  Scalar components of $X_k$ are written
$(X_k)_j$ or $X_{k,j}$.

\section{Preliminaries~\label{sec:prelim}}
This section collects the essential tools used to establish the bounds.  For background on information-theoretic quantities (mutual information, KL divergence, entropy), we refer the reader
to~\cite{CoverThomas2006}.

\subsection{Sub-Gaussian Random Vectors and Concentration Inequalities}
\begin{definition}[$\sigma$-Sub-Gaussian Random Variable; {\cite[Definition~2.2]{Wainwright2019HDS}}]~\label{def:subgaussianrv}
    A random variable $X$ is $\sigma$-sub-Gaussian if
    \[
        \mathbb{E}\big[e^{\lambda(X-\mathbb{E}[X])}\big]\leq e^{\lambda^2\sigma^2/2},\quad\forall\lambda\in\mathbb{R}.
    \]
\end{definition}
\begin{definition}[$\sigma$-Sub-Gaussian Random Vector; {\cite[Section~6.3]{Wainwright2019HDS}}]~\label{def:subgaussianvec}
    A random vector $Z\in\mathbb{R}^d$ is $\sigma$-sub-Gaussian if $u^TZ$ is a $\sigma$-sub-Gaussian random variable for every $u\in\mathbb{S}^{d-1}$, where $\mathbb{S}^{d-1}=\{u\in\mathbb{R}^d:\|u\|_2=1\}$. We denote the family of all $d$-dimensional $\sigma$-sub-Gaussian distributions by $\mathrm{subG}^{(d)}(\sigma)$.
\end{definition}
\begin{remark}
    For a Gaussian vector $Z\sim N(0,C)$, we have
    $N(0,C)\in\mathrm{subG}^{(d)}(\sigma)$ if and only if
    $\|C\|_{op}\leq\sigma^2$. This equivalence is used
    repeatedly in the hypothesis family constructions.
\end{remark}
\begin{definition}[Sub-Exponential Random Variable; {\cite[Definition~2.8.4]{vershynin2026high}}]~\label{def:subexponential}
    A random variable $X$ is sub-exponential with parameter $K$ if
    \[
        \|X\|_{\psi_1}:=\inf\{t>0:\mathbb{E}[e^{|X|/t}]\leq 2\}\leq K.
    \]
    The quantity $\|X\|_{\psi_1}$ is called the $\psi_1$ (sub-exponential) Orlicz norm.
\end{definition}
\begin{lemma}[Product of Sub-Gaussians; {\cite[Lemma~2.8.6]{vershynin2026high}}]~\label{lem:subgaussianproduct}
    If $X$ is $\sigma_1$-sub-Gaussian and $Y$ is $\sigma_2$-sub-Gaussian, then $(X-\mathbb{E}X)(Y-\mathbb{E}Y)$ is sub-exponential with $\|(X-\mathbb{E}X)(Y-\mathbb{E}Y)\|_{\psi_1}\leq C\sigma_1\sigma_2$, where $C>0$ is a universal constant.
\end{lemma}
\begin{lemma}[Moment Bound for Sub-Exponentials; {\cite[Proposition~2.8.1(ii)]{vershynin2026high}}]~\label{lem:subexpmoment}
    If $X$ is sub-exponential, then $\mathbb{E}[X^2]\leq C\|X\|_{\psi_1}^2$ for a universal constant $C>0$.
\end{lemma}

\begin{theorem}[Bernstein's Inequality, Theorem 2.9.1 in \cite{vershynin2026high}]~\label{thm:bernstein}
    Let $X_1,\dots,X_N$ be independent, mean zero, sub-exponential random variables. Then, for every $t\geq 0$, we have
    $$\mathbb{P}\bigg\{\bigg|\sum_{i=1}^N X_i\bigg|\geq t\bigg\}\leq 2\exp\big[-c\min(\frac{t^2}{\sum_{i=1}^N\|X_i\|^2_{\psi_1}},\frac{t}{\max_i\|X_i\|_{\psi_1}})\big]$$
\end{theorem}
% \begin{lemma}[Concentration of Empirical Cross-Covariance; {\cite[Lemma~F.1]{rahmani2025fundamental}}]~\label{lem:rahmanif1}
%     Assume that $\mathbf{X}\in\mathbb{R}^{d_1}$ is a zero mean, sub-Gaussian vector with parameter $\sigma_1$, and we have $m$ i.i.d.\ samples from $\mathbf{X}$ as $\{\mathbf{X}^{(i)}\}_{i=1}^m$. Also assume that $\mathbf{Y}\in\mathbb{R}^{d_2}$ is a zero mean, sub-Gaussian vector with parameter $\sigma_2$, and we have $m$ i.i.d.\ samples from $\mathbf{Y}$ as $\{\mathbf{Y}^{(i)}\}_{i=1}^m$. Consider the cross-covariance matrix $\mathbf{C_{XY}}\in\mathbb{R}^{d_1\times d_2}$ as $\mathbf{C_{XY}}=\mathbb{E}[\mathbf{X}\mathbf{Y}^T]$ and assume that we use the estimator $\tilde{\mathbf{C}}_{\mathbf{XY}}=\frac{1}{m}\sum_{i=1}^m\mathbf{X}^{(i)}\mathbf{Y}^{(i)T}$. Then we have:
%     \[
%         \mathbb{P}\Big[\|\tilde{\mathbf{C}}_{\mathbf{XY}}-\mathbf{C_{XY}}\|_{op}\geq 10\sigma_1\sigma_2 t\Big]\leq (9)^{d_1+d_2}\exp\big(-m\cdot\min\{t,t^2\}\big),
%     \]
%     and:
%     \[
%         \mathbb{P}\Big[\|\tilde{\mathbf{C}}_{\mathbf{XY}}\|_{op}\geq 11\sigma_1\sigma_2\Big]\leq\min\big\{1,\,\exp\big(3(d_1+d_2)-m\big)\big\}.
%     \]
% \end{lemma}
\begin{lemma}[Concentration of Empirical Cross-Covariance;
{\cite[Lemma~F.1]{rahmani2025fundamental}}]\label{lem:f1crosscov-conc}
Let $U\in\mathbb{R}^{p_1}$ and $V\in\mathbb{R}^{p_2}$ be zero-mean
random vectors that are $\sigma_1$- and $\sigma_2$-sub-Gaussian
respectively.  Given $m$ i.i.d.\ copies
$(U^{(1)},V^{(1)}),\dots,(U^{(m)},V^{(m)})$, define the empirical
cross-covariance
$\widetilde{C}_{UV}:=\frac{1}{m}\sum_{i=1}^m U^{(i)}{V^{(i)}}^\top$
and the true cross-covariance
$C_{UV}:=\mathbb{E}[UV^\top]$.  Then for $t>0$,
\[
    \mathbb{P}\Big[\big\|\widetilde{C}_{UV}-C_{UV}
    \big\|_{\mathrm{op}}\geq 10\,\sigma_1\sigma_2\, t\Big]
    \leq 9^{p_1+p_2}\,
    \exp\!\big({-m\cdot\min\{t,\,t^2\}}\big),
\]
and $\mathbb{P}\big[\|\widetilde{C}_{UV}\|_{\mathrm{op}}
\geq 11\,\sigma_1\sigma_2\big]
\leq\min\big\{1,\,\exp\!\big(3(p_1+p_2)-m\big)\big\}$.
\end{lemma}
\begin{proposition}[Expected Operator Norm Error of Empirical
Cross-Covariance;
\textup{\cite[Proposition~F.2]{rahmani2025fundamental}}]
\label{prop:f2crosscov-expected}
Let $U\in\mathbb{R}^{p_1}$ and $V\in\mathbb{R}^{p_2}$ be
zero-mean random vectors that are $\sigma_1$- and
$\sigma_2$-sub-Gaussian respectively.  Given $m$ i.i.d.\
copies $(U^{(1)},V^{(1)}),\dots,(U^{(m)},V^{(m)})$, define
the empirical cross-covariance
$\widetilde{C}_{UV}:=\frac{1}{m}\sum_{i=1}^m
U^{(i)}{V^{(i)}}^\top$ and the true cross-covariance
$C_{UV}:=\mathbb{E}[UV^\top]$.  Then
\[
    \mathbb{E}\big[\|\widetilde{C}_{UV}-C_{UV}
    \|_{\mathrm{op}}\big]
    \;\leq\; 32\,\sigma_1\sigma_2
    \max\!\left\{\sqrt{\frac{p_1+p_2}{m}},\;
    \frac{p_1+p_2}{m}\right\}.
\]
\end{proposition}
% \begin{lemma}[Operator Norm of Sub-Gaussian Random Matrix; {\cite[Lemma~F.3]{rahmani2025fundamental}}]~\label{lem:rahmanif3}
%     Let $\mathbf{A}$ be a $d\times n$ random matrix whose columns $\mathbf{A}_i$ are independent, mean zero, $\sigma$-sub-Gaussian random vectors. Then we have:
%     \[
%         \mathbb{P}\Big[\|\mathbf{A}\|_{op}\geq 6\sigma\sqrt{d+n}\Big]\leq\exp\big(-2(d+n)\big).
%     \]
%     Moreover, for $q\in\{1,2\}$ the following inequality holds:
%     \begin{equation}~\label{eq:rahmanif3moment}
%         \mathbb{E}\Big[\|\mathbf{A}\|_{op}^q\Big]\leq C_q\sigma^q(d+n)^{q/2},
%     \end{equation}
%     for some universal constant $C_q$ depending only on $q$.
% \end{lemma}
\begin{lemma}[Operator Norm of Sub-Gaussian Random Matrix;
{\cite[Lemma~F.3]{rahmani2025fundamental}}]\label{lem:f3subG-opnorm}
Let $\Gamma\in\mathbb{R}^{p_1\times p_2}$ be a random matrix whose
columns $\Gamma_1,\dots,\Gamma_{p_2}$ are independent, zero-mean,
$\sigma$-sub-Gaussian random vectors in $\mathbb{R}^{p_1}$.  Then
\[
    \mathbb{P}\Big[\|\Gamma\|_{\mathrm{op}}
    \geq 6\,\sigma\sqrt{p_1+p_2}\Big]
    \leq\exp\!\big({-2(p_1+p_2)}\big).
\]
Moreover, for $q\in\{1,2\}$,
\begin{equation}\label{eq:subG-opnorm-moment}
    \mathbb{E}\Big[\|\Gamma\|_{\mathrm{op}}^q\Big]
    \leq C_q\,\sigma^q\,(p_1+p_2)^{q/2},
\end{equation}
where $C_q>0$ is a universal constant depending only on $q$.
\end{lemma}

\subsection{Conditional Strong Data Processing Inequality (Conditional SDPI)}
The standard SDPI quantifies the contraction of KL divergence through
a Markov kernel: the SDPI coefficient
$s(P_X,T_{Y|X})\in[0,1]$ satisfies
$D_{\mathrm{KL}}(Q_Y\|P_Y)\leq s(P_X,T_{Y|X})\cdot
D_{\mathrm{KL}}(Q_X\|P_X)$ for all $Q_X\ll
P_X$, and is studied by~\cite{AG76,raginsky2016strong} amongst others.  In the vertical split, however, the channel from $X_1$ to $X_2$ depends on the unknown cross-covariance $C_{21}$, which plays the role of a latent state~$V$. This state-dependence means the standard SDPI does not directly apply. \cite{rahmani2025fundamental} introduced the \emph{Conditional SDPI}, which averages the contraction over the state
distribution.

To state it in the form used in this paper, consider a standardised
model where $X_1\sim N(0,I_{d_1})$ and
$X_2\mid(X_1,V\!=\!v)\sim N(A_vX_1,\;I_{d_2}-A_vA_v^\top)$,
so that the conditional relationship is the Gaussian channel
\begin{equation}\label{eq:channel}
    X_2 \;=\; A_v\,X_1 + Z_v,
    \qquad
    Z_v\sim N(0,\,I_{d_2}-A_vA_v^\top)
    \;\perp\!\!\!\perp\; X_1,
\end{equation}
where $V$ is a random variable indexing the cross-covariance
$A_v\in\mathbb{R}^{d_2\times d_1}$ with
$A_vA_v^\top\preceq I_{d_2}$, and $X_1\perp\!\!\!\perp V$.

\begin{theorem}[C-SDPI for Gaussian Mixtures;
{\cite[Theorem~3.3]{rahmani2025fundamental}}]\label{thm:csdpi}
Under the channel model~\eqref{eq:channel}, the conditional SDPI
coefficient equals
\begin{equation}\label{eq:csdpicoeff}
    s\big(P_{X_1},\,T_{X_2|X_1,V}
    \,\big|\,P_V\big)
    \;=\; \big\|\mathbb{E}\big[A_V^\top A_V\big]\big\|_{\mathrm{op}}.
\end{equation}
Moreover, the C-SDPI coefficient tensorises: for the full dataset
$\bX_1=(X_1^{(1)},\dots,X_1^{(m)})$ passed through $m$ independent
copies of the same channel with a fixed state~$V$,
\begin{equation}\label{eq:tensorise}
    s\big(P_{\bX_1},\,T_{\bX_2|\bX_1,V}
    \,\big|\,P_V\big)
    \;=\;
    s\big(P_{X_1},\,T_{X_2|X_1,V}
    \,\big|\,P_V\big)
    \;=\;
    \big\|\mathbb{E}\big[A_V^\top A_V\big]\big\|_{\mathrm{op}},
\end{equation}
independent of~$m$~\cite[Theorem~3.1]{rahmani2025fundamental}.
\end{theorem}

\subsection{Linear Algebra, Statistics and Information Theory Facts}
\begin{lemma}[Schur Complement]~\label{lem:schur}
    Let $M=\begin{pmatrix}A&B\\C&D\end{pmatrix}$ with $A$ square and invertible.
    \begin{enumerate}[label=(\roman*)]
        \item Positive definiteness: If $M$ is symmetric (so in particular $C=B^T$), then $M\succ 0$ if and only if $A\succ 0$ and $D-B^TA^{-1}B\succ 0$.
        \item Determinant factorization: $\det(M)=\det(A)\cdot\det(D-CA^{-1}B)$.
    \end{enumerate}
\end{lemma}
\begin{lemma}[Eigenvalues of Symmetric Block Matrices; {\cite[Lemma~A.1]{rahmani2025fundamental}}]~\label{lem:blockeig}
    Let $A\in\mathbb{R}^{m\times n}$ with singular value decomposition $A=\sum_{i=1}^r\sigma_iu_iv_i^T$. Then the matrix
    \[
        B=\begin{pmatrix}0&A\\A^T&0\end{pmatrix}\in\mathbb{R}^{(m+n)\times(m+n)}
    \]
    has eigenvalues $\{\pm\sigma_i\}_{i=1}^r$ and $0$ with multiplicity $m+n-2r$.
\end{lemma}
\begin{lemma}[Gaussian Conditioning Formula]~\label{lem:gaussiancond}
    Let $(X_1,X_2)\sim N(0,C)$ with $C=\begin{pmatrix}C_{11}&C_{12}\\C_{21}&C_{22}\end{pmatrix}$ and $C_{11}\succ 0$. Then
    \[
        X_2\mid X_1\sim N\big(C_{21}C_{11}^{-1}X_1,\;C_{22}-C_{21}C_{11}^{-1}C_{12}\big).
    \]
\end{lemma}

\begin{theorem}[Entropy of a Multivariate Normal Distribution;
{\cite[Theorem~8.4.1]{CoverThomas2006}}]
    \label{thm:gaussianentropy}
Let $X_1,X_2,\dots,X_n$ have a multivariate normal distribution
with mean $\mu$ and positive definite covariance matrix $K$.  Then
    $$h(X_1,X_2,\dots,X_n)
    = h(N_n(\mu,K))
    = \frac{1}{2}\log\big((2\pi e)^n|K|\big),$$
where $|K|$ denotes the determinant of $K$.
\end{theorem}
\subsection{Signed Permutation Matrices}
\begin{definition}[Partial signed permutation matrix]
A \emph{partial signed permutation matrix} of size $s'$ in
$\R^{d_2 \times d_1}$ is a matrix
\begin{equation}\label{eq:pp-def}
    E = \sum_{k=1}^{s'} \eta_ke_{j_k} e_{i_k}^T,\hspace{5mm}\eta_k\in\{-1,+1\}
\end{equation}
where $i_1, \ldots, i_{s'} \in [d_1]$ are distinct and
$j_1, \ldots, j_{s'} \in [d_2]$ are distinct. Such a matrix has
entries in $\{-1,0,+1\}$, exactly $s'$ nonzeros, and at most one nonzero
per row and per column. We require $s' \leq \dmin := d_1 \wedge d_2$.
\end{definition}

\begin{lemma}[Lemma 2.2 in \cite{rahmani2025fundamental}]~\label{lem:signedpermmat}
    Let $A$ be a random matrix drawn uniformly from the set of signed permutation matrices $\mathcal{P}_d$. Then, for any fixed matrix $B\in\mathbb{R}^{d\times d}$, the following holds:
    \begin{equation}~\label{eq:signedpermmat}
        \mathbb{E}[A^TBA]=\frac{\text{Tr}\{B\}}{d}I_d.
    \end{equation}
\end{lemma}

\subsection{Vectorization of Matrices}
\begin{definition}[Vectorization]~\label{def:vectorisation}
The bijection $\vecop: \R^{d_2 \times d_1} \to \R^{d_1 d_2}$ stacks
the columns of a matrix into a single vector. For any
$E, E' \in \R^{d_2 \times d_1}$:
\begin{equation}\label{eq:vec-isom}
    \|E - E'\|_F = \|\vecop(E) - \vecop(E')\|_2.
\end{equation}
This is the Frobenius--$\ell_2$ isometry
(This follows from {\cite[Appendix~A.6]{rahmani2025fundamental}}).
\end{definition}

\subsection{Packing, Covering, and Quantization~\label{sec:packing}}
\begin{definition}[Packing and Covering Numbers; {\cite[Definitions~5.1, 5.4]{Wainwright2019HDS}}]~\label{def:packingcovering}
    Let $(T,\rho)$ be a metric space.
    \begin{enumerate}[label=(\roman*)]
        \item A $\delta$-\emph{packing} of $T$ is a set $\{x_1,\dots,x_N\}\subset T$ such that $\rho(x_i,x_j)\geq\delta$ for all $i\neq j$. The \emph{packing number} $M(\delta;T,\rho)$ is the largest cardinality of a $\delta$-packing.
        \item A $\delta$-\emph{covering} (or $\delta$-\emph{net}) of $T$ is a set $\{x_1,\dots,x_N\}\subset T$ such that for every $x\in T$, there exists $x_i$ with $\rho(x,x_i)\leq\delta$. The \emph{covering number} $N(\delta;T,\rho)$ is the smallest cardinality of a $\delta$-covering.
    \end{enumerate}
\end{definition}
\begin{lemma}[Matrix Quantization via Covering Nets; {\cite[Appendix~A.6.1]{rahmani2025fundamental}}]~\label{lem:quantization}
    Let $A\in\mathbb{R}^{p\times q}$ with $\|A\|_{op}\leq r$. Using an $\omega$-covering of the operator-norm ball $\mathcal{B}^{p\times q}_{\|\cdot\|_{op}}(r)$ encoded with $B$ bits, one can produce a quantized matrix $\hat{A}$ satisfying
    \[
        \|\hat{A}-A\|_{op}\leq\omega,\qquad\text{where }\omega=3r\cdot 2^{-B/(pq)},
    \]
    provided $B\geq pq\log_2(3r/\omega)$.
\end{lemma}
% ============================================================
%  §3.2 Fano's Inequality and the Averaged Fano Method
%  Detailed version — replaces the corresponding subsection
%  in sections23_v2.tex
%
%  Notation: non-bold X_k = single sample, bold \bX_k = dataset
% ============================================================

\subsection{Fano's Inequality and the Averaged Fano Method}%
\label{sec:fano}

We now recall Fano's method for minimax lower bounds and develop the two specialisations used in this paper: a basic version (without auxiliary randomness) used for the cross-covariance sample lower bound in Appendix~\ref{app:crosssample}, and an \emph{averaged} version (with auxiliary randomness $W$) used for the cross-covariance communication lower bound.

\subsubsection*{General Setup}

Let $\mathcal{P}$ be a family of distributions, and let $\theta\colon\mathcal{P}\to\Theta$ be a parameter of interest residing in a metric space $(\Theta,\|\cdot\|)$ (in our case, $\Theta$ is the space of covariance matrices or cross-covariance matrices equipped with the Frobenius norm). The goal is to approximate $\theta(P)$ for an unknown $P\in\mathcal{P}$ from data $X\sim P$.

A finite subset $\mathcal{P}_{\mathcal{V}}=\{P_v\}_{v\in\mathcal{V}}\subset\mathcal{P}$ is called \emph{$2\delta$-separated} if $\|\theta(P_v)-\theta(P_{v'})\|\geq 2\delta$ for all $v\neq v'\in\mathcal{V}$.  Fano's method reduces the estimation problem to a hypothesis testing problem: if no test can reliably distinguish which $P_v$ generated the data, then no estimator can achieve error smaller than~$\delta$.

\subsubsection*{Standard Fano Inequality}

\begin{lemma}[Fano's Inequality;
{\cite[Corollary~15.12]{Wainwright2019HDS}}]\label{lem:fano}
Let $\mathcal{P}_{\mathcal{V}}=\{P_v\}_{v\in\mathcal{V}} \subset\mathcal{P}$ be $2\delta$-separated. Let $V$ be drawn uniformly from $\mathcal{V}$, and given $V=v$, let $X\sim P_v$.  Assume $|\mathcal{V}|\geq 2$. Then
\begin{equation}\label{eq:fano}
    \inf_{\hat{\theta}}\;\sup_{P\in\mathcal{P}}\;
    \mathbb{E}_P\big[\|\hat{\theta}(X)-\theta(P)\|\big]
    \;\geq\; \delta\left(1-\frac{I(V;\,X)+\log 2}
    {\log|\mathcal{V}|}\right).
\end{equation}
\end{lemma}

\noindent We refer the reader to \cite{Wainwright2019HDS} for the proof.

\subsubsection*{Application to the DCME/DCCME Setting (Without $W$)}
% To apply Fano's inequality to our distributed setting, we need to relate the mutual information $I(V;X)$ in~\eqref{eq:fano} to the communication constraints $B_1,B_2$.  In the DCME setting, the data $X$ consists of the pair $(\bX_1,\bX_2)$ (the full datasets of both agents), and the estimator $\hat{\theta}$ has access to these only through the encoded messages $M_1=\mathcal{E}_1(\bX_1)$ and $M_2=\mathcal{E}_2(\bX_2)$.  The key observation is the Markov chain
% \begin{equation}\label{eq:markov}
%     V \;\longrightarrow\; (\bX_1,\bX_2)
%     \;\longrightarrow\; (M_1,M_2)
%     \;\longrightarrow\; \hat{C},
% \end{equation}
% which gives, by the data processing inequality, $I(V;\hat{C})\leq I(V;\,M_1,M_2)\leq I(V;\,\bX_1,\bX_2)$.Substituting into~\eqref{eq:fano}, we obtain the following.
Similar to \cite{rahmani2025fundamental}, the following variant of Fano's Inequality is utilized:
\begin{lemma}[Fano for DCME/DCCME]\label{lem:fanodcme}
Let $\mathcal{P}_{\mathcal{V}}=\{P_v\}_{v\in\mathcal{V}}
\subset\mathcal{P}_s\subset\mathrm{subG}^{(d)}(\sigma)$ be a family of $|\mathcal{V}|$ distributions with covariance matrices $C_v$ and cross-covariance matrices $D_v:=C_{v,21}$. Define the pairwise separations
\[
    \rho_F
    := \inf_{\substack{(v,v')\in\mathcal{V}^2\\v\neq v'}}
    \|C_v-C_{v'}\|_F,
    \qquad
    \rho_F^{(\mathrm{cross})}
    := \inf_{\substack{(v,v')\in\mathcal{V}^2\\v\neq v'}}
    \|D_v-D_{v'}\|_F.
\]
Let $V$ be uniform on $\mathcal{V}$, and given $V=v$, let $\bZ=\{Z^{(i)}\}_{i=1}^m\overset{\mathrm{i.i.d.}}{\sim}P_v$, with agents observing $\bX_1$ and $\bX_2$ respectively.  Fix any DCME (respectively, DCCME) scheme with parameters $(\sigma,m,d_{1:2},B_{1:2})$, and let $M_k=\mathcal{E}_k(\mathbf{X}_k)$, $k=1,2$. Assume $|\mathcal{V}|\geq 2$. We have
\begin{align}
     \sup_{P\in\mathcal{P}_s}
\mathbb{E}\!\left[\big\|\hat{C}-C\big\|_F \right] &\;\geq\; \frac{\rho_F}{2}
    \left(1-\frac{I(V;\,M_1,M_2)+\log 2}
    {\log|\mathcal{V}|}\right),
    \label{eq:fanodcme}\\[4pt]
    \sup_{P\in\mathcal{P}_s}
\mathbb{E}\!\left[\big\|\hat{C}_{21}-C_{21}\big\|_F\right] &\;\geq\;
    \frac{\rho_F^{(\mathrm{cross})}}{2}
    \left(1-\frac{I(V;\,M_1,M_2)+\log 2}
    {\log|\mathcal{V}|}\right).
    \label{eq:fanodccme}
\end{align}
\end{lemma}

\begin{remark}\label{rem:fanodcme}
The bound in Lemma~\ref{lem:fanodcme} reduces the minimax lower bound problem to two tasks: \emph{(i)}~constructing a well-separated family $\{P_v\}_{v\in\mathcal{V}}$ (large $\rho_F^{(\mathrm{cross})}$ and large $|\mathcal{V}|$), and \emph{(ii)}~upper-bounding the mutual information $I(V;\,M_1,M_2)$ uniformly over all admissible schemes. After substituting such a uniform information bound into Eqs.~\ref{eq:fanodcme} and~\ref{eq:fanodccme}, we may take the infimum over the encoders and decoder to obtain the corresponding minimax lower bounds.
\end{remark}

\subsubsection*{Averaged Fano Inequality (With Auxiliary Randomness $W$)}
In the general $s$-sparse setting, the hypothesis family is parameterised not only by the index $v\in\mathcal{V}$ (which selects the perturbation direction) but also by an auxiliary variable $w\in\mathcal{W}$ (which randomizes the perturbation by left and right signed permutations).  This auxiliary randomisation is needed to make the C-SDPI coefficient computation tractable: averaging over $W$ makes the expected Gram matrices scalar multiples of the identity.

To handle this, we use the \emph{Averaged Fano} method of~\cite{rahmani2025fundamental}, which extends Fano's inequality to families that vary with~$W$.
\begin{lemma}[Averaged Fano;
{\cite[Example~15.19]{Wainwright2019HDS},
\cite[Lemma~5.1]{rahmani2025fundamental}}]
\label{lem:avgfano}
Let $W\sim\pi_W$ be an auxiliary random variable taking values in $\mathcal{W}$.  For each $w\in\mathcal{W}$, let $\mathcal{P}_{\mathcal{V}}^{(w)}=\{P_v^{(w)}\}_{v\in\mathcal{V}}
\subset\mathcal{P}_s\subset\mathrm{subG}^{(d)}(\sigma)$ be a family of $|\mathcal{V}|$ distributions with covariance matrices $C_v^{(w)}$ and cross-covariance matrices $D_v^{(w)}$.  Define the worst-case separations
\[
    \rho_F
    := \inf_{\substack{w\in\mathcal{W}\\(v,v')\in\mathcal{V}^2,\,
    v\neq v'}}
    \big\|C_v^{(w)}-C_{v'}^{(w)}\big\|_F,
    \qquad
    \rho_F^{(\mathrm{cross})}
    := \inf_{\substack{w\in\mathcal{W}\\(v,v')\in\mathcal{V}^2,\,
    v\neq v'}}
    \big\|D_v^{(w)}-D_{v'}^{(w)}\big\|_F.
\]
Let $V$ be uniform on $\mathcal{V}$, independent of $W$, and given $(W,V)=(w,v)$, let $\bZ=\{Z^{(i)}\}_{i=1}^m\overset{\mathrm{i.i.d.}}{\sim}P_v^{(w)}$, with agents observing $\bX_1$ and $\bX_2$ respectively.  Fix any DCME (respectively, DCCME) scheme with parameters $(\sigma,m,d_{1:2},B_{1:2})$, and let $M_k=\mathcal{E}_k(\mathbf{X}_k)$, $k=1,2$. Assume $|\mathcal{V}|\geq 2$. We have
\begin{align}
    \sup_{P\in\mathcal{P}_s}
\mathbb{E}\!\left[\|\hat{C}-C\|_F\right]&\;\geq\; \frac{\rho_F}{2}
    \left(1-\frac{I(W,V;\,M_1,M_2)+\log 2}
    {\log|\mathcal{V}|}\right),
    \label{eq:avgfanodcme}\\[4pt]
    \sup_{P\in\mathcal{P}_s}
\mathbb{E}\!\left[\|\hat{C}_{21}-C_{21}\|_F\right] &\;\geq\;
    \frac{\rho_F^{(\mathrm{cross})}}{2}
    \left(1-\frac{I(W,V;\,M_1,M_2)+\log 2}
    {\log|\mathcal{V}|}\right).
    \label{eq:avgfanodccme}
\end{align}
\end{lemma}

\section{Main Results~\label{sec:mainresults}}
\subsection{Lower Bounds\label{sec:lowerbounds}}
Throughout this section, let $s':=s\wedge d_{\min}$ and $L:=\frac{3}{4}\log(2\,d_1 d_2)-\log s'-\log(4e)$. We have the asymptotic bound:
\begin{corollary}~\label{cor:lasymp}
    There exists a universal constant $D>0$ such that whenever $\frac{d_1d_2}{s'}\ge D$, we have
    \[ L=\Theta\bigg(\log\frac{d_1d_2}{s'}\bigg).\]
\end{corollary}
\begin{proof}
    Since $s'\le d_{\min} \le \sqrt{d_1d_2}$, we have $\log s'\le\frac12\log d_1d_2$. Also,
    \[ L = \frac34\log(2d_1d_2)-\log s'-\log(4e) = \frac34 \log d_1d_2 - \log s' -\underbrace{\big(\frac54 + \log e\big)}_{:=c_0}.\]
    Note that
    \[ \log\frac{d_1d_2}{s'}-L = \log d_1d_2 - \log s' - \big(\frac34\log d_1d_2 - \log s' - c_0\big) = \frac14\log d_1d_2 + c_0>0.\]
    Hence,
    \[
    L\le\log\frac{d_1d_2}{s'}.
    \]
    For the reverse bound, by $\log s'\le\frac12\log d_1d_2$, we have
    \[
    L-\frac12\log\frac{d_1d_2}{s'} = \frac14\log d_1d_2 - \frac12 \log s' - c_0 \ge -c_0.
    \]
    Hence,
    \[
    L\ge\frac12\log\frac{d_1d_2}{s'}-c_0.
    \]
    If $\log\frac{d_1d_2}{s'}\ge4c_0$, then $c_0\le\frac14\log\frac{d_1d_2}{s'}$, and so
    \[
    L\ge\frac14\log\frac{d_1d_2}{s'}.
    \]
    Since $2^{4c_0}=2^{5+4\log e} = 32e^4$, one such possible constant is
    \[
    \frac{d_1d_2}{s'}\ge 32e^4.
    \]
\end{proof}
\begin{remark}
The numerical threshold in Corollary~4.1 is conservative, constants of a much smaller order would also work. Moreover, since $s'\le d_{\min}$ implies $d_1d_2/s'\ge d_{\max}$, the logarithmic comparison holds eventually whenever $d_{\max}\to\infty$, uniformly over all admissible choices of $s'$.
\end{remark}

Our first result gives a minimax lower bound for cross-covariance estimation under the Frobenius norm.

\begin{theorem}[$s$-sparse DCCME Lower Bound]
\label{thm:ssparse-dccme}
Consider the $\mathrm{DCCME}(\sigma,m,d_{1:2},B_{1:2})$ problem
with $C_{21}\in\mathcal{S}(s)$.   Suppose $s'L\ge4$. Then
\begin{equation}\label{eq:ssparse-dccme}
    \MF^{(\mathrm{cross})} \;\geq\; \frac{\sigma^2}{32}\left[
    \left(\sqrt{s'L\left(\frac{d_1}{B_1}
    \vee\frac{d_2}{B_2}\right)}
    \;\vee\;\sqrt{\frac{s'L}{m}}\right)
    \;\wedge\;\sqrt{s'}\right].
\end{equation}
In particular, for large enough $d_1d_2/s'$ such that $L=\Theta\big(\log\frac{d_1d_2}{s'}\big)$,
\begin{equation}\label{eq:ssparse-dccme-asymp}
    \MF^{(\mathrm{cross})} \;=\;
    \Omega\!\left(\sigma^2\left[\left(
    \sqrt{s'\log\!\left(\frac{d_1d_2}{s'}\right)
    \!\left(\frac{d_1}{B_1}\vee\frac{d_2}{B_2}\right)}
    \;\vee\;
    \sqrt{\frac{s'\log(d_1d_2/s')}{m}}
    \right)\land\sqrt{s'}\right]\right).
\end{equation}
\end{theorem}
\noindent

\begin{corollary}\label{cor:ssparse-dccme}
Under the assumptions of Theorem~\ref{thm:ssparse-dccme}, suppose
\[0<\varepsilon<\frac{\sigma^2\sqrt{s'}}{32}.\]
Then any $s$-sparse $DCCME$ scheme achieving expected Frobenius distortion $\varepsilon$ must satisfy, for $k=1,2$,
\[B_k\ge\frac{\sigma^4d_k s'L}{2^{10}\varepsilon^2}\quad\text{ and }\quad m\ge\frac{\sigma^4s'L}{2^{10}\varepsilon^2}.\]
Consequently, for large enough $d_1d_2/s'$ such that $L=\Theta\big(\log\frac{d_1d_2}{s'}\big)$, we have
\begin{equation}\label{eq:ssparse-dccme-cor}
    B_k \;=\; \Omega\!\left(
    \frac{\sigma^4\,d_k\,s'\log(d_1d_2/s')}
    {\varepsilon^2}\right),
    \qquad
    m \;=\; \Omega\!\left(
    \frac{\sigma^4\,s'\log(d_1d_2/s')}
    {\varepsilon^2}\right).
\end{equation}
\end{corollary}
\begin{proof}
    We have
    \[
    \frac{\sigma^2}{32}\left[
    \left(\sqrt{s'L\left(\frac{d_1}{B_1}
    \vee\frac{d_2}{B_2}\right)}
    \;\vee\;\sqrt{\frac{s'L}{m}}\right)
    \;\wedge\;\sqrt{s'}\right] \le\varepsilon<\frac{\sigma^2\sqrt{s'}}{32},
    \]
    so
    \[
    \sqrt{s'L\left(\frac{d_1}{B_1}
    \vee\frac{d_2}{B_2}\right)}
    \;\vee\;\sqrt{\frac{s'L}{m}}<\sqrt{s'} \implies \frac{\sigma^2}{32}\left(\sqrt{s'L\left(\frac{d_1}{B_1}
    \vee\frac{d_2}{B_2}\right)}
    \;\vee\;\sqrt{\frac{s'L}{m}}\right) \le\varepsilon.
    \]
    Then,
    \[
    \sqrt{s'L\left(\frac{d_1}{B_1}
    \vee\frac{d_2}{B_2}\right)}
    \;\vee\;\sqrt{\frac{s'L}{m}}\le\frac{32\varepsilon}{\sigma^2}\implies \sqrt{\frac{s'L d_k}{B_k}}\le\frac{32\varepsilon}{\sigma^2}\implies B_k \ge \frac{\sigma^4 d_ks'L}{2^{10}\varepsilon^2},\text{ for }k=1,2.
    \]
    Similarly,
    \[
    \sqrt{s'L\left(\frac{d_1}{B_1}
    \vee\frac{d_2}{B_2}\right)}
    \;\vee\;\sqrt{\frac{s'L}{m}}\le\frac{32\varepsilon}{\sigma^2}\implies \sqrt{\frac{s'L}{m}}\le\frac{32\varepsilon}{\sigma^2}\implies m\ge\frac{\sigma^4s'L}{2^{10}\varepsilon^2}.
    \]
\end{proof}

\noindent Our second result extends the cross-covariance bound to the full covariance matrix by incorporating the self-covariance terms.

\begin{theorem}[$s$-Sparse DCME Lower Bound]
\label{thm:ssparse-dcme}
Under the same assumptions as Theorem~\ref{thm:ssparse-dccme},
\begin{align}\label{eq:ssparse-dcme}
    \MF &\;\geq\; \frac{\sigma^2}{32}\left[
    \left(\sqrt{s'L\left(\frac{d_1}{B_1}\vee
    \frac{d_2}{B_2}\right)}
    \;\vee\;\sqrt{\frac{s'L}{m}}\right)
    \;\wedge\;\sqrt{s'}\right]
    \nonumber\\[4pt]
    &\quad\;\vee\;\;
    \frac{\sigma^2}{7168}\left(
    \frac{d_1\vee d_2}{\sqrt{m}}
    \;\wedge\;\sqrt{2(d_1\vee d_2)}\right)
    \nonumber\\[4pt]
    &\quad\;\vee\;\;
    \frac{\sigma^2}{56}\left(
    \sqrt{d_1}\cdot 2^{-16B_1/d_1^2}
    \;\vee\;
    \sqrt{d_2}\cdot 2^{-16B_2/d_2^2}\right).
\end{align}
In particular, for large enough $d_1d_2/s'$ such that $L=\Theta\big(\log\frac{d_1d_2}{s'}\big)$,
% \begin{equation}\label{eq:ssparse-dcme-asymp}
%     \MF \;=\; \Omega\!\left(\sigma^2\left(
%     \sqrt{s'\log\!\left(\frac{d_1d_2}{s'}\right)
%     \!\left(\frac{d_1}{B_1}\vee\frac{d_2}{B_2}\right)}
%     \;\vee\;
%     \frac{d_1\vee d_2}{\sqrt{m}}
%     \;\vee\;
%     \sqrt{d_1\cdot 2^{-16B_1/d_1^2}}
%     \;\vee\;
%     \sqrt{d_2\cdot 2^{-16B_2/d_2^2}}
%     \right)\right).
% \end{equation}
\begin{align}\label{eq:ssparse-dcme-asymp}
\MF
&=
\Omega\!\Bigg(
\sigma^2
\Bigg[
\left(
\left(
\sqrt{
s'\log\!\left(\frac{d_1d_2}{s'}\right)
\left(
\frac{d_1}{B_1}\vee\frac{d_2}{B_2}
\right)
}
\vee
\sqrt{
\frac{
s'\log(d_1d_2/s')
}{m}
}
\right)
\wedge
\sqrt{s'}
\right)
\nonumber\\
& \vee
\left(
\frac{d_1\vee d_2}{\sqrt{m}}
\wedge
\sqrt{d_1\vee d_2}
\right)
\vee
\sqrt{
d_1}\,2^{-16B_1/d_1^2
}
\vee
\sqrt{
d_2}\,2^{-16B_2/d_2^2
}
\Bigg]
\Bigg).
\end{align}
\end{theorem}
% \begin{corollary}\label{cor:ssparse-dcme}
% Any $s$-sparse DCME scheme achieving expected Frobenius distortion
% $\varepsilon$ requires
% \begin{equation}\label{eq:ssparse-dcme-cor}
%     B_k \;=\; \Omega\!\left(
%     \frac{\sigma^4\,d_k\,s'\log(d_1d_2/s')}
%     {\varepsilon^2}\right),
%     \qquad
%     m \;=\; \Omega\!\left(
%     \frac{\sigma^4\,(d_1\vee d_2)^2}
%     {\varepsilon^2}\right).
% \end{equation}
% \end{corollary}
\begin{corollary}\label{cor:ssparse-dcme}
Under the assumptions of Theorem~\ref{thm:ssparse-dcme}, suppose
\[
0<\varepsilon<
\sigma^2\left(
\frac{\sqrt{s'}}{32}
\wedge
\frac{\sqrt{2(d_1\vee d_2)}}{7168}
\right).
\]
Then any $s$-sparse $DCME$ scheme achieving expected Frobenius distortion $\varepsilon$ must satisfy, for $k=1,2$,
\[
B_k\ge\frac{\sigma^4d_k s'L}{2^{10}\varepsilon^2}
\quad\text{and}\quad
m\ge\frac{\sigma^4(d_1\vee d_2)^2}{7168^2\varepsilon^2}.
\]
Consequently, for large enough $d_1d_2/s'$ such that
$L=\Theta\big(\log\frac{d_1d_2}{s'}\big)$, we have
\begin{equation}\label{eq:ssparse-dcme-cor}
    B_k \;=\; \Omega\!\left(
    \frac{\sigma^4\,d_k\,s'\log(d_1d_2/s')}
    {\varepsilon^2}\right),
    \qquad
    m \;=\; \Omega\!\left(
    \frac{\sigma^4\,(d_1\vee d_2)^2}
    {\varepsilon^2}\right).
\end{equation}
\end{corollary}

\begin{remark}[Special 1-sparse case]
\label{rem:specialcases}
Setting $s=1$ gives $s'=1$ and with the same standing assumptions, Theorem~\ref{thm:ssparse-dccme} reduces to
\begin{equation}~\label{eq:1sparse-dccme}
    \MF^{(\mathrm{cross})} \;=\;
    \Omega\!\left(\sigma^2\left[\left(
    \sqrt{\log(d_1d_2)
    \left(\frac{d_1}{B_1}\vee\frac{d_2}{B_2}\right)}
    \;\vee\;
    \sqrt{\frac{\log(d_1 d_2)}{m}}\right)
    \land 1\right]\right).
\end{equation}
\end{remark}

\subsection{Achievability of Bounds}
We construct a DCME protocol that matches the communication lower bounds of Theorems~\ref{thm:ssparse-dccme} and \ref{thm:ssparse-dcme} up to polylogarithmic factors under the conditions of Remark~\ref{rem:achievmatch}.  The protocol combines the covering-net quantisation framework of~\cite{rahmani2025fundamental} with entry-wise hard thresholding inspired by~\cite{garg2014communication}.

\begin{theorem}[$s$-Sparse Achievability]
\label{thm:sachievability}
Consider the $\mathrm{DCME}(\sigma,m,d_{1:2},B_{1:2})$
problem with $C_{21}\in\mathcal{S}(s)$.  Let $d=d_1+d_2$
and $0<\tilde\varepsilon=\varepsilon/(\sigma^2\sqrt{d})\leq 1$.
Suppose $d_1d_2\ge2$,
\begin{align}
m &\ge2
\max\left\{
    \left\lceil\frac{2^{19}d}{\tilde\varepsilon^2}\right\rceil,
    \left\lceil
        C_2\left(
            d+\frac{\sigma^4s\log(d_1d_2)}{\varepsilon^2}
        \right)
    \right\rceil
\right\},
\label{eq:achiev-m}\\[3pt]
B_k &\ge
1+\beta d_k
\left\lceil
    C_2\left(
        d+\frac{\sigma^4s\log(d_1d_2)}{\varepsilon^2}
    \right)
\right\rceil
+
\left\lceil
    2d_k^2\log\!\left(\frac{528}{\tilde\varepsilon}\right)
\right\rceil,
\label{eq:achiev-B}\\
\intertext{where}
\beta &=
\max\left\{
    2,
    \left\lceil
        \log\!\left(
            \frac{C_3\sigma^2\sqrt{s}}{\varepsilon}
        \right)
    \right\rceil
\right\},
\notag
\end{align}
and $C_2,C_3$ are universal constants chosen in the proof.  Then there exists a DCME protocol achieving
$\mathbb{E}[\|\hat{C}-C\|_F]\leq\varepsilon$.
\end{theorem}
The factor 2 in Eq.~\ref{eq:achiev-m} accounts for paired differencing in Step 1 in Section 6. The extra bit in Eq.~\ref{eq:achiev-B} is for an error flag.
% \label{rem:achievmatch}
% The cross-covariance component of the achievable distortion satisfies
% \[
%     \mathbb{E}\big[\|\hat{C}_{21}^{\,\mathrm{thr}}
%     -C_{21}\|_F\big]
%     \;\leq\;
%     C\,\sigma^2\sqrt{s\log(d_1 d_2)
%     \left(\frac{d_1}{B_1}\vee\frac{d_2}{B_2}\right)}
%     \cdot\sqrt{\beta},
% \]
% which matches the cross-covariance lower bound of Proposition~\ref{prop:ssparse-comm} up to the
% polylogarithmic factor
% $\sqrt{\beta}=\sqrt{2\log(C_3\sigma^2/\varepsilon)}$.
\begin{remark}[Matching the lower bound]\label{rem:achievmatch}
Under the hypotheses for the logarithmic-rate lower bound in
Corollary~\ref{cor:ssparse-dccme}, suppose
\[
    s\le d_{\min},\qquad
    0<\varepsilon<\frac{\sigma^2\sqrt{s}}{32},
    \qquad
    \varepsilon^2\le
    \frac{\sigma^4s\log(d_1d_2)}{d}.
\]
Then $s'=s$, and the cross-covariance communication budget
can be chosen to satisfy
\[
    B_k^{(\mathrm{cov})}
    =O\!\left(
        \frac{\beta\sigma^4d_k s\log(d_1d_2)}
             {\varepsilon^2}
    \right).
\]
Since $s\le d_{\min}\le\sqrt{d_1d_2}$,
$\log(d_1d_2)=\Theta(\log(d_1d_2/s))$.
Consequently, under existing assumptions, the cross-covariance communication requirement
matches the lower bound up to the logarithmic factor
\[
    \beta
    =O\!\left(\log\frac{\sigma^2\sqrt{s}}{\varepsilon}\right).
\]
Under these assumptions, $d_k^2\le d_kd\le\sigma^4d_ks\log(d_1d_2)/\varepsilon^2$, so the self-covariance communication cost is also absorbed up to logarithmic factors, and the total communication budget matches the lower bound up to polylogarithmic factors.
\end{remark}

\section{Proof of Lower Bounds~\label{sec:prooflower}}
\subsection{Proof Architecture~\label{sec:lowerproofarchi}}
The proof of Theorem~\ref{thm:ssparse-dcme} establishes four independent lower bounds, each isolating a different source of estimation difficulty:
\begin{center}
\begin{tabular}{@{}lll@{}}
\toprule
Section & Source of difficulty & Rate \\
\midrule
\ref{sec:cross-comm}
  & Communication for $C_{21}$
  & $\frac{\sigma^2}{32}
\left[
\sqrt{
s'L
\left(
\frac{d_1}{B_1}
\vee
\frac{d_2}{B_2}
\right)
}
\wedge
\sqrt{s'}
\right]$ \\
\ref{app:crosssample}
  & Sample for $C_{21}$
  & $\frac{\sigma^2}{32}
\left[
\sqrt{
\frac{s'L}{m}
}
\wedge
\sqrt{s'}
\right]$ \\
\ref{app:fullsample}
  & Sample for $C_{kk}$
  & $\frac{\sigma^2}{7168}
\left[
\frac{d_1\vee d_2}{\sqrt{m}}
\wedge
\sqrt{2(d_1\vee d_2)}
\right]$ \\
\ref{app:self-comm}
  & Communication for $C_{kk}$
  & $\frac{\sigma^2}{56}
\left[
\sqrt{
d_1}\,2^{-16B_1/d_1^2
}
\vee
\sqrt{
d_2}\,2^{-16B_2/d_2^2
}
\right]$ \\
\bottomrule
\end{tabular}
\end{center}
Recall that by Corollary~\ref{cor:lasymp}, under certain conditions, we have $L=\Theta\big(\log\frac{d_1d_2}{s'}\big)$.

\subsection{Results on Sparse Packing}
The following lemma extends Lemma 4 of \cite{raskutti2011minimax} to the partial permutation setting. As it is the combinatorial core of the $s$-sparse lower bound, we state the lemma and its proof in detail here.

\begin{lemma}~\label{lem:VGpartialperm}
    For $1\le s'\le d_1\wedge d_2$, define $\mathcal{G}(s')$ as the set of vectors in $\{-1,0,+1\}^{d_1d_2}$ with exactly $s'$ nonzero entries whose unique matrix representation via $\vecop^{-1}$, defined in Definition~\ref{def:vectorisation}, is a signed partial
permutation.
    % $$\mathcal{G}(s'):=\{s'\text{-sparse vectors in }\{-1,0,+1\}^{d_1d_2} \text{ with unique matrix rep. via }\vecop\text{ a signed partial permutation}\}.$$
    There exists a subset
    $\tilde{\mathcal{G}}(s') \subset \mathcal{G}(s')$ with cardinality
    \[
        |\tilde{\mathcal{G}}(s')| \geq
        \bigg(\frac{(2d_1d_2)^{3/4}}{4es'}\bigg)^{s'},
    \]
    such that $\rho_H(\vecop(E),\vecop(E'))>s'/2$ for all $E\neq E'$ with $\vecop(E),\vecop(E')\in\tilde{\mathcal{G}}(s')$, where $\rho_H$ denotes Hamming distance.
\end{lemma}
\begin{proof}
    We adapt the greedy construction argument in \cite{raskutti2011minimax} Appendix~A. Within this proof, we identify each element of $\mathcal{G}(s')$ with its matrix representation under $\vecop^{-1}$, and write $\rho_H(E,E')$ for the number of entries at which the matrices $E$ and $E'$ differ.

    Each element of $\mathcal{G}(s')$ is specified by choosing
    $s'$ column indices from $[d_1]$, then assigning each a
    distinct row index from $[d_2]$, then choosing a sign
    $\pm 1$ for each entry. Hence, by $\binom{n}{k}\ge\frac{n^k}{k^k}$ and $s'!\ge\big(\frac{s'}{e}\big)^{s'}$, we have
    \begin{equation}\label{eq:ground-set-size}
    |\mathcal{G}(s')| = \binom{d_1}{s'}\cdot
    \frac{d_2!}{(d_2 - s')!}\cdot 2^{s'}
    \geq \left(\frac{d_1}{s'}\right)^{s'}
    \cdot\left(\frac{d_2}{s'}\right)^{s'}\cdot\left(\frac{s'}{e}\right)^{s'}\cdot 2^{s'}
    = \left(\frac{2\,d_1 d_2}{es'}\right)^{s'}.
    \end{equation}

    For some fixed $E\in\mathcal{G}(s')$, consider the set $\{E'\in\mathcal{G}(s'):\rho_H(\vecop(E),\vecop(E'))\le s'/2\}$. We wish to upper-bound its cardinality.

    Since both $E$ and $E'$ have exactly $s'$ nonzeros, write
    $c := |S(E) \setminus S(E')| = |S(E') \setminus S(E)|$
    (entries present in one but not the other) and
    $b := |\{(j,i) \in S(E) \cap S(E') : E_{ji} \neq E'_{ji}\}|$
    (shared positions with different signs). The Hamming distance
    decomposes as $\rho_H = b + 2c$, so
    $\rho_H \leq s'/2$ implies $c \leq s'/4$.

    To construct such an $E'$ from $E$:
    \begin{enumerate}
        \item \emph{Choose which entries to remove $(S(E)\setminus S(E'))$.}
        Select $c$ of $E$'s $s'$ support positions to remove.
        Summing over $c = 0, 1, \ldots, \lfloor s'/4\rfloor$:
        \[
            \sum_{c=0}^{\lfloor s'/4 \rfloor}\binom{s'}{c}
            < 2^{s'} \text{ choices.}
        \]
    \item \emph{Place the $c$ new entries $(S(E')\setminus S(E))$.}
        Each new entry needs a position $(j', i')$ in the
        $d_2 \times d_1$ matrix and a sign $\pm 1$.
        Crudely (ignoring the row/column occupancy constraint):
        at most $2 d_1 d_2$ choices per entry, giving
        $(2\,d_1 d_2)^c \leq (2\,d_1 d_2)^{s'/4}$ total.
    \item \emph{Choose sign flips at shared positions.}
        Each of the $s' - c$ kept entries can independently
        keep or flip its sign: at most $2^{s'-c} \leq 2^{s'}$
        choices.
    \end{enumerate}
Hence:
\begin{equation}\label{eq:ball-bound}
    |\{E' \in \mathcal{G}(s') : \rho_H \leq s'/2\}|
    < 2^{s'}\cdot(2\,d_1 d_2)^{s'/4}\cdot 2^{s'}
    = 4^{s'}\cdot(2\,d_1 d_2)^{s'/4}.
\end{equation}
% Note that the first inequality is strict because the crude bound $(2d_1d_2)^c$ strictly overcounts.

Consider any $\emptyset\neq \mathcal{A} \subset \mathcal{G}(s')$ with
$|\mathcal{A}| \leq M$, where
$M := |\mathcal{G}(s')|\big/
\big(4^{s'}\cdot(2\,d_1 d_2)^{s'/4}\big)$.
The set of elements $E\in\mathcal{G}(s')$ that are within Hamming distance $s'/2$ of some element of $\mathcal{A}$ has cardinality at most
\[
\begin{aligned}
|\{E\in\mathcal{G}(s')\mid\rho_H(E,E')\le s'/2\text{ for some }E'\in\mathcal{A}\}|&<|\mathcal{A}|\cdot 4^{s'}\cdot (2d_1d_2)^{s'/4}\\
&\le M\cdot 4^{s'}\cdot (2d_1d_2)^{s'/4}=|\mathcal{G}(s')|.
\end{aligned}
\]

Consequently, for any such set with cardinality $|\mathcal{A}|\le M$, there exists an $E\in\mathcal{G}(s')$ such that $\rho_H(E,E')>s'/2$ for all $E'\in\mathcal{A}$. By inductively adding this element at each round, we then construct a subset $\mathcal{A}\subset\mathcal{G}(s')$ with $|\mathcal{A}|> M$ such that $\rho_H(E,E')>s'/2$ for all distinct $E,E'\in\mathcal{A}$.

To conclude, let us lower-bound the quantity $M$. We have
\begin{align*}
    M&=\frac{\binom{d_1}{s'}\cdot\frac{d_2!}{(d_2-s')!}\cdot 2^{s'}}{4^{s'}\cdot (2d_1d_2)^{s'/4}}\ge\frac{(\frac{2d_1d_2}{es'})^{s'}}{4^{s'}\cdot (2d_1d_2)^{s'/4}}=\frac{(2d_1d_2)^{s'}}{e^{s'}\cdot s'^{s'}\cdot 4^{s'}\cdot(2d_1d_2)^{s'/4}}=\bigg(\frac{(2d_1d_2)^{3/4}}{4es'}\bigg)^{s'}
\end{align*}
\end{proof}

We present another crucial lemma bridging the Hamming distance used in Lemma~\ref{lem:VGpartialperm} to the Frobenius norm.

\begin{lemma}[Frobenius--Hamming inequality for
$\{-1,0,+1\}$-valued matrices]
\label{lem:frob-hamming}
For any $E, E' \in \{-1,0,+1\}^{d_2 \times d_1}$:
\[
    \|E - E'\|_F^2 \geq
    \rho_H(\vecop(E),\, \vecop(E')).
\]
Equality holds when $E, E'$ have disjoint supports.
\end{lemma}

\begin{proof}
At each position $(j,i)$, both $E_{ji}$ and $E'_{ji}$
belong to $\{-1,0,+1\}$. If $E_{ji} = E'_{ji}$, the
contribution to both sides is 0. If $E_{ji} \neq E'_{ji}$,
the Hamming side contributes 1, and the Frobenius side
contributes $(E_{ji} - E'_{ji})^2$. The possible
disagreements are:
\begin{center}
\renewcommand{\arraystretch}{1.2}
\begin{tabular}{cc|cc}
$E_{ji}$ & $E'_{ji}$ & $(E_{ji}-E'_{ji})^2$
    & $\mathbb{I}[E_{ji}\neq E'_{ji}]$ \\
\hline
$+1$ & $0$ & $1$ & $1$ \\
$-1$ & $0$ & $1$ & $1$ \\
$0$ & $+1$ & $1$ & $1$ \\
$0$ & $-1$ & $1$ & $1$ \\
$+1$ & $-1$ & $4$ & $1$ \\
$-1$ & $+1$ & $4$ & $1$ \\
\end{tabular}
\end{center}
In every case, $(E_{ji}-E'_{ji})^2 \geq
\mathbb{I}[E_{ji}\neq E'_{ji}]$. Summing over all
$(j,i)$, we have
\[
    \|E - E'\|_F^2
    = \sum_{j,i}(E_{ji}-E'_{ji})^2
    \geq \sum_{j,i}\mathbb{I}[E_{ji}\neq E'_{ji}]
    = \rho_H(\vecop(E),\vecop(E')).
\]
Equality holds when every disagreement is of the $\pm 1$ vs $0$ type (contribution $1 = 1$), which occurs when $S(E) \cap S(E') = \emptyset$ (disjoint supports).
\end{proof}

\subsection{The $s$-Sparse Hypothesis Family}
We chain together the vectorization isometry as in Definition~\ref{def:vectorisation} and
Lemma~\ref{lem:VGpartialperm} to construct the $s$-sparse
hypothesis family.

Let $s' := s \wedge \dmin$. By Lemma~\ref{lem:VGpartialperm},
there exists a subset
$\tilde{\mathcal{G}} \subset \mathcal{G}(s')$ with pairwise
Hamming separation $> s'/2$ and
$|\tilde{\mathcal{G}}| \geq
\left(\frac{(2d_1 d_2)^{3/4}}{4es'}\right)^{s'}$. 
Enumerate its elements as
$z_1, z_2, \ldots, z_{|\tilde{\mathcal{G}}|}$. Define the index set
$\mathcal{V} = [1 : |\tilde{\mathcal{G}}|]$, and define
the \emph{base perturbation matrices}
\begin{equation}\label{eq:base-Ev}
    E_v := \vecop^{-1}(z_v) \in
    \{-1,0,+1\}^{d_2 \times d_1},
    \qquad v \in \mathcal{V},
\end{equation}
which are signed partial permutations of size $s'$
by definition of $\mathcal{G}(s')$.

% Let $W \sim \pi_W$ be a random variable taking values
% in a finite set $\mathcal{W}$, independent of $V$.
% For each $w \in \mathcal{W}$, consider a family
% of distributions
% $\mathcal{P}_{\mathcal{V}}^{(w)}
% = \{P_v^{(w)}\}_{v \in \mathcal{V}}$, where
% $P_v^{(w)} = N(0, C_v^{(w)})$, and
% \begin{equation}\label{eq:ssparse-Cvw}
%     C_v^{(w)} = \frac{\sigma^2}{2}
%     \begin{pmatrix}
%         I_{d_1} & \delta\,(E_v^{(w)})^T \\
%         \delta\, E_v^{(w)} & I_{d_2}
%     \end{pmatrix},
% \end{equation}
% where $E_v^{(w)}$ is a signed partial permutation of
% size $s'$ derived from the base matrix $E_v$, with
% $\|E_v^{(w)}\|_{\op} \leq 1$, and $\delta \in (0,1]$
% is a parameter to be determined subsequently.
% The specific construction of $E_v^{(w)}$ from $E_v$
% via $W$ is deferred to next step. 
% Clearly, $|(C_v^{(w)})_{21}|_0 \leq s,
% \forall v \in \mathcal{V}, w \in \mathcal{W}$, since $(C_v^{(w)})_{21} = \frac{\sigma^2}{2}\delta E_v^{(w)}$ and $|E_v^{(w)}|_0 = s' = s \wedge \dmin \leq s$.

Let $\mathcal P_{d_k}$ denote the set of
$d_k\times d_k$ signed permutation matrices, for $k=1,2$.
Let
\[
    W=(W_L,W_R)
    \sim
    \operatorname{Unif}(\mathcal P_{d_2})
    \otimes
    \operatorname{Unif}(\mathcal P_{d_1}),
\]
independent of $V$, and define
    $\mathcal W
    :=
    \mathcal P_{d_2}\times\mathcal P_{d_1}.$
For each $w=(w_L,w_R)\in\mathcal W$, define
\begin{equation}\label{eq:ssparse-randomized-Ev}
    E_v^{(w)}
    :=
    w_LE_vw_R.
\end{equation}
For each $w\in\mathcal W$, consider a family of distributions $\mathcal P_{\mathcal V}^{(w)}=\{P_v^{(w)}\}_{v\in\mathcal V},$ where $P_v^{(w)}=N(0,C_v^{(w)})$, and
\begin{equation}\label{eq:ssparse-Cvw}
    C_v^{(w)}
    =
    \frac{\sigma^2}{2}
    \begin{pmatrix}
        I_{d_1}
        &
        \delta\,(E_v^{(w)})^T
        \\
        \delta\,E_v^{(w)}
        &
        I_{d_2}
    \end{pmatrix},
\end{equation}
where $\delta\in(0,1/\sqrt{2}]$ is a parameter to be
determined subsequently.

Since left and right multiplication by signed permutation
matrices only permutes rows and columns and possibly flips signs,
$E_v^{(w)}$ is again a signed partial permutation of size $s'$.
In particular,
\[
    |E_v^{(w)}|_0=s',
    \qquad
    \|E_v^{(w)}\|_{\op}
    =
    \|E_v\|_{\op}
    =1.
\]
Clearly,
\[
    |(C_v^{(w)})_{21}|_0\le s,
    \qquad
    \forall\,v\in\mathcal V,\ w\in\mathcal W,
\]
since $(C_v^{(w)})_{21}=\frac{\sigma^2}{2}\delta E_v^{(w)}$
and
$|E_v^{(w)}|_0=s'=s\wedge\dmin\le s$.

It is routine to check that $C_v^{(w)}$ is positive definite, and that the corresponding distribution $P_v^{(w)}$ is $\sigma$-sub-Gaussian. Combined with the sparsity check in the last paragraph, the construction is valid. Moreover, we have the separation $\rho_F = \sqrt{2}\rho_F^{(\mathrm{cross})}.$ We defer these checks to Appendix~\ref{app:verification}.

Since the marginal covariance matrices in
\eqref{eq:ssparse-Cvw} are fixed,
the marginal distributions of $\mathbf X_1$ and $\mathbf X_2$
do not depend on $(V,W)$. Hence $\mathbf X_k\perp(V,W),
    \qquad k=1,2.$
Since $M_k$ is a function of $\mathbf X_k$, we also have
$M_k\perp(V,W)$. This, combined with the chain rule of mutual
information, gives
\begin{align}
    I(V,W;M_1,M_2)
    &=
    I(V,W;M_1)
    +
    I(V,W;M_2\mid M_1)
    \nonumber\\
    &=
    I(V,W;M_2\mid M_1)
    \qquad
    (M_1\perp(V,W))
    \nonumber\\
    &\le
    I(V,W,M_1;M_2)
    \nonumber\\
    &=
    I(M_1;M_2\mid V,W)
    \qquad
    (M_2\perp(V,W)).
    \label{eq:chainrule}
\end{align}
Moreover, by Lemma~\ref{lem:VGpartialperm}, the base
perturbation matrices $\{E_v\}_{v\in\mathcal V}$ satisfy
\[
    \rho_H\big(
    \vecop(E_v),\vecop(E_{v'})
    \big)
    >
    \frac{s'}{2},
    \qquad
    v\neq v'.
\]
By Lemma~\ref{lem:frob-hamming},
\[
    \inf_{v\neq v'}
    \|E_v-E_{v'}\|_F
    \ge
    \sqrt{\frac{s'}{2}}.
\]
Since signed permutation matrices are orthogonal, the Frobenius
norm is invariant under left and right multiplication. Hence,
for every $w=(w_L,w_R)\in\mathcal W$,
\[
\begin{aligned}
    \|E_v^{(w)}-E_{v'}^{(w)}\|_F
    &=
    \|w_L(E_v-E_{v'})w_R\|_F\\
    &=
    \|E_v-E_{v'}\|_F,
\end{aligned}
\]
so the separation is $w$-independent.
Furthermore, we have
\begin{equation}\label{eq:ssparse-packing-number}
    \log|\mathcal V|
    \ge
    s'
    \log
    \frac{(2d_1d_2)^{3/4}}{4es'}
    =
    s'\underbrace{\left(\frac34\log(2d_1d_2)-\log s'-\log(4e)\right)}_{:=L},
\end{equation}
with cross-covariance separation satisfying
\begin{equation}\label{eq:ssparse-rho-value}
    \rho_F^{(\mathrm{cross})}
    \ge
    \frac{\sigma^2}{2}
    \delta
    \sqrt{\frac{s'}{2}}
    =
    \frac{\sigma^2\delta\sqrt{s'}}
    {2\sqrt{2}}.
\end{equation}

\begin{remark}
    We give a brief remark on the construction for the special 1-sparse case to aid with understanding. In this special case, we define the
index set
\[
    \mathcal{V} = \{(i,j) : i \in [d_1],\, j \in [d_2]\},
    \qquad |\mathcal{V}| = d_1 d_2,
\]
and for each $v = (i,j) \in \mathcal{V}$, the perturbation matrix
$E_v = e_j e_i^T \in \mathbb{R}^{d_2 \times d_1}$, the distribution
$P_v = N(0, C_v)$, and the covariance
\begin{equation}\label{eq:1sparsehypofam}
    C_v := \frac{\sigma^2}{2}
    \begin{pmatrix}
        I_{d_1} & \delta\, E_v^T \\
        \delta\, E_v & I_{d_2}
    \end{pmatrix},
\end{equation}
where $e_i \in \mathbb{R}^{d_1}$ and $e_j \in \mathbb{R}^{d_2}$
are standard basis vectors, and $\delta \in (0,1/\sqrt{2}]$ is a free parameter to be optimized. The same validity properties as above hold. The packing cardinality however is much simpler: The hypothesis set $\{E_v\}_{v \in \mathcal{V}}
= \{e_j e_i^T : i \in [d_1],\, j \in [d_2]\}$ contains
$|\mathcal{V}| = d_1 d_2$ matrices. Each satisfies the feasibility
condition $\|E_v\|_{\mathrm{op}} = 1 \leq 1$ (since $e_j e_i^T$ is a
rank-1 matrix with a single singular value equal to 1). For
distinct $v = (i,j) \neq v' = (i',j')$, the matrices $e_j e_i^T$
and $e_{j'} e_{i'}^T$ have disjoint support, so:
\begin{equation}\label{eq:Ev-separation}
    \|E_v - E_{v'}\|_F^2
    = \|e_j e_i^T\|_F^2 + \|e_{j'} e_{i'}^T\|_F^2
    = 1 + 1 = 2.
\end{equation}
Hence $\{E_v\}$ is a $\sqrt{2}$-packing of
$\{A \in \mathbb{R}^{d_2 \times d_1} : |A|_0 \leq 1,\,
\|A\|_{\mathrm{op}} \leq 1\}$ under the Frobenius norm, and $\log|\mathcal{V}| = \log(d_1 d_2)$.\\
The reason the packing for this case is a lot simpler is that the canonical basis matrices
$\{e_je_i^\top\}_{(j,i)\in[d_2]\times[d_1]}$ form a natural packing of size $d_1d_2$ that is trivially well-separated: any two distinct basis matrices have disjoint supports, hence no probabilistic argument is needed.
\end{remark}
\begin{remark}[Comparison with the dense case]
\cite{rahmani2025fundamental} pack the operator-norm unit ball $\{D \in \mathbb{R}^{d_2 \times d_1} : \|D\|_{\mathrm{op}} \leq 1\}$ with \emph{dense} matrices, obtaining $\log|\mathcal{V}| \geq 2\,d_1 d_2$ via a volumetric argument (\cite[Appendix A.6]{rahmani2025fundamental}).
\end{remark}

\subsection{Communication Lower Bound for $C_{21}$~\label{sec:cross-comm}}
\subsubsection*{Applying the C-SDPI Theorem}
Conditioned on $(W,V)=(w,v)$, the Markov chain $M_1\to \mathbf{X_1} \to \mathbf{X_2}\to M_2$ holds. By the data-processing inequality, we have
\begin{equation}~\label{eq:1sparsedpi}
    I(M_1;M_2|V,W)\le I(M_1;\mathbf{X_2}|V,W)\land I(M_2;\mathbf{X_1}|V,W)
\end{equation}

Conditionally on $(V,W)=(v,w)$, Lemma~\ref{lem:gaussiancond} and $(X_1,X_2)\sim N(0,C_v^{(w)})$ imply $X_2=\delta E_v^{(w)}X_1+Z_v^{(w)}$, where $Z_v^{(w)}\sim N(0,\Sigma_v^{(w)})$ is independent of $X_1$ under this conditional law, with $\Sigma_v^{(w)}=C_{22}-C_{21}C_{11}^{-1}C_{12}$ and $C_{21}C_{11}^{-1}=\delta E_v^{(w)}$, where $C_{ij}$ denote the blocks of $C_v^{(w)}$. Theorem~\ref{thm:csdpi} is stated for a \emph{single} sample pair $(X_1,X_2)$ with identity marginal covariances; multiplying both coordinates by $\sqrt{2}/\sigma$ gives this normalization without changing mutual information or the C-SDPI coefficient. However, conditionally on $(V,W)$, each agent observes $m$ i.i.d.\ samples: Agent 1 sees $\mathbf{X}_1=(X_1^{(1)},\ldots,X_1^{(m)})$ and encodes all of them into a single message $M_1$. The relevant channel is therefore the $m$-fold product channel $\mathbf{X}_1\to\mathbf{X}_2$ conditional on the common state $(V,W)$. Applying the tensorization property with state $(V,W)$, we have
\begin{align}
    I(M_1;\mathbf{X_2}|V,W)\le\delta^2\bigg\|\mathbb{E}_{(W,V)}\big[(E_V^{(W)})^TE_V^{(W)}\big]\bigg\|_{\text{op}}I(M_1;\mathbf{X_1})\le\delta^2\bigg\|\mathbb{E}_{(W,V)}\big[(E_V^{(W)})^TE_V^{(W)}\big]\bigg\|_{\text{op}}B_1
\end{align}
\begin{align}
    I(M_2;\mathbf{X_1}|V,W)\le\delta^2\bigg\|\mathbb{E}_{(W,V)}\big[E_V^{(W)}(E_V^{(W)})^T\big]\bigg\|_{\text{op}}I(M_2;\mathbf{X_2})\le\delta^2\bigg\|\mathbb{E}_{(W,V)}\big[E_V^{(W)}(E_V^{(W)})^T\big]\bigg\|_{\text{op}}B_2
\end{align}

\subsection*{Evaluating the C-SDPI Constant via Signed Permutation Matrices}
Recall that $W=(W_L,W_R)\sim\text{Unif}(\mathcal{P}_{d_2})\otimes\text{Unif}(\mathcal{P}_{d_1})$, and $E_v^{(W)} = W_L E_vW_R$. Since signed permutation matrices are orthogonal, we have
\[W_L^TW_L=I_{d_2},\quad W_RW_R^T=I_{d_1}.\]
Therefore,
\begin{align}
    (E_v^{(W)})^T E_v^{(W)}
    &=
    W_R^T E_v^T W_L^T W_L E_v W_R
    =
    W_R^T E_v^T E_v W_R,
    \label{eq:ssparse-right-gram}
\end{align}
whereas
\begin{align}
    E_v^{(W)}(E_v^{(W)})^T
    &=
    W_L E_v W_R W_R^T E_v^T W_L^T
    =
    W_L E_v E_v^T W_L^T.
    \label{eq:ssparse-left-gram}
\end{align}
By Lemma~\ref{lem:signedpermmat},
\begin{align}
    \mathbb E_W
    \left[
        (E_v^{(W)})^T E_v^{(W)}
    \right]
    &=
    \mathbb E_{W_R}
    \left[
        W_R^T E_v^T E_v W_R
    \right]
    \nonumber\\
    &=
    \frac{1}{d_1}
    \operatorname{Tr}(E_v^T E_v)
    I_{d_1}
    \nonumber\\
    &=
    \frac{s'}{d_1}I_{d_1}.
    \label{eq:ssparse-gram1}
\end{align}
For the other direction, note that the transpose map is a
bijection of the set of signed permutation matrices, so
$W_L^T$ is also uniformly distributed over
$\mathcal P_{d_2}$. Hence Lemma~\ref{lem:signedpermmat}
also gives
\begin{align}
    \mathbb E_W
    \left[
        E_v^{(W)}(E_v^{(W)})^T
    \right]
    &=
    \mathbb E_{W_L}
    \left[
        W_L E_v E_v^T W_L^T
    \right]
    \nonumber\\
    &=
    \frac{1}{d_2}
    \operatorname{Tr}(E_vE_v^T)
    I_{d_2}
    \nonumber\\
    &=
    \frac{s'}{d_2}I_{d_2}.
    \label{eq:ssparse-gram2}
\end{align}
Here we used
\[
    \operatorname{Tr}(E_v^TE_v)
    =
    \operatorname{Tr}(E_vE_v^T)
    =
    \|E_v\|_F^2
    =
    s',
\]
since every $E_v$ is a signed partial permutation of size $s'$.

\noindent The identities \eqref{eq:ssparse-gram1} and \eqref{eq:ssparse-gram2} hold for every $v\in\mathcal V$. Therefore, averaging additionally over $V$ gives
\begin{equation}\label{eq:ssparse-gram-averages}
    \mathbb E_{W,V}
    \left[
        (E_V^{(W)})^T E_V^{(W)}
    \right]
    =
    \frac{s'}{d_1}I_{d_1},
    \qquad
    \mathbb E_{W,V}
    \left[
        E_V^{(W)}(E_V^{(W)})^T
    \right]
    =
    \frac{s'}{d_2}I_{d_2}.
\end{equation}
Hence, combining the above, we have
\begin{align}
    I(V,W;M_1,M_2)
    &\le
    I(M_1;M_2\mid V,W)
    \nonumber\\
    &\le
    I(M_1;\mathbf X_2\mid V,W)
    \wedge
    I(M_2;\mathbf X_1\mid V,W)
    \nonumber\\
    &\le
    \delta^2
    \left(
        \frac{s'}{d_1}B_1
        \wedge
        \frac{s'}{d_2}B_2
    \right)
    \nonumber\\
    &=
    \delta^2s'
    \left(
        \frac{B_1}{d_1}
        \wedge
        \frac{B_2}{d_2}
    \right).
    \label{eq:ssparse-mutualinfo}
\end{align}

\subsection*{Fano Bound Assembly}
Substituting Eqs.~\ref{eq:ssparse-packing-number},~\ref{eq:ssparse-rho-value}, and~\ref{eq:ssparse-mutualinfo} into Eq.~\ref{eq:avgfanodccme}, we have
\begin{equation}~\label{eq:ssparse-fano-raw}
    \MF^{\text{(cross)}}\ge\frac{1}{2}\cdot\frac{\sigma^2\delta\sqrt{s'}}
    {2\sqrt{2}}\cdot\bigg(1-\frac{\delta^2s'\bigg(\frac{B_1}{d_1}\wedge \frac{B_2}{d_2}\bigg)+1}{s'L}\bigg),
\end{equation}
where $L := \frac{3}{4}\log(2\,d_1 d_2) - \log s' - \log(4e)$. Following typical information-theoretic arguments, we choose $\delta$ such that
$\frac{I(V,W;M_1,M_2)}{\log|\mathcal{V}|}
\leq\frac{1}{2}$, for which it suffices to require
$$
    \frac{s'\delta^2(B_1/d_1 \wedge B_2/d_2)}{s'L}
    = \frac{\delta^2(B_1/d_1 \wedge B_2/d_2)}{L}
    \leq \frac{1}{2}
    \implies \delta^2 \leq \frac{L}{2}
    \left(\frac{d_1}{B_1}\vee\frac{d_2}{B_2}\right).
$$
\noindent Recalling the constraint $\delta \leq \frac{1}{\sqrt{2}}$, choose:
\begin{equation}\label{eq:ssparse-delta}
    \delta^2 = \bigg[\frac{L}{2}
    \bigg(\frac{d_1}{B_1}\vee\frac{d_2}{B_2}\bigg)
    \bigg] \wedge \frac12.
\end{equation}
Therefore, we have the bound
\begin{proposition}[$s$-Sparse Cross Communication Lower Bound]
\label{prop:ssparse-comm}
Under the assumptions of Lemma~\ref{lem:VGpartialperm}
and the additional condition $s'L \geq 4$
(both of which hold when $d_1 d_2/s'$ is sufficiently
large), we have
\begin{equation}\label{eq:ssparse-comm-final}
    \MFcross \geq \frac{\sigma^2}{32}
    \bigg(\sqrt{s' L\bigg(\frac{d_1}{B_1}
    \vee\frac{d_2}{B_2}\bigg)}
    \;\wedge\; \sqrt{s'}\bigg),
\end{equation}
where $L = \frac{3}{4}\log(2d_1 d_2)
- \log s' - \log(4e)$.

\noindent Moreover, when $d_1d_2/s'$ is sufficiently large so that $L=\Theta\left(\log\frac{d_1d_2}{s'}\right)$, we have
\begin{equation}
    \MFcross = \Omega\bigg(\sigma^2\left[
    \sqrt{s'\log\!\bigg(\frac{d_1 d_2}{s'}\bigg)
    \!\bigg(\frac{d_1}{B_1}\vee\frac{d_2}{B_2}\bigg)}
    \land\sqrt{s'}\right]\bigg).
\end{equation}
\end{proposition}

\begin{remark}
    As seen in Section~\ref{sec:lowerproofarchi}, there exist additional lower bounds pertaining to sample complexity and the limited communication budget for self-covariance estimation. The proofs for these lower bounds are presented in Appendix~\ref{app:lowerbounds}.
\end{remark}

\section{Achievability Protocol and Proof Sketch~\label{sec:proofachiev}}
In this section, we propose a protocol and derive upper bounds on both its sample complexity and communication budgets for approximation with accuracy within $\varepsilon$ of the covariance matrix under the $\|\cdot\|_F$ norm. We use the protocol proposed in Section 6 of \citep{rahmani2025fundamental} as the backbone in Steps 1--4, with a thresholding element from \citep{garg2014communication} in Step 5.

\subsubsection*{Step 1: Mean Invariance}
Each agent locally computes $X^{'(i)}_k=\frac{X^{(2i-1)}_k-X_k^{(2i)}}{\sqrt{2}}, i=1,\dots,\lfloor m/2\rfloor, \quad k=1,2$~\label{sec:protocolstep1}. The agents use the same pair indices and discard an unpaired final observation (if any). The resulting joint observations are independent, mean-zero, $\sigma-$sub-Gaussian and have the original covariance $C$. For notational simplicity, from this point in Section 6 and throughout Appendix~\ref{app:achievfull}, relabel $\lfloor m/2\rfloor$ as $m$ and the differenced observations $X_k^{'(i)}$ as $X_k^{(i)}$. Thus $m$ below is the number of available mean-zero observations. The factor 2 in Theorem~\ref{thm:sachievability} already accounts for this relabeling.

\subsubsection*{Step 2: Bit Allocation~\label{sec:protocolstep2}}
\noindent Agent $k$ reserves one error-flag bit and splits the remaining budget into
\[
B_k^{(\text{self})}=\left\lceil 2d_k^2\log\bigg(\frac{528}{\tilde{\varepsilon}}\bigg)\right\rceil,\qquad B_k^{(\mathrm{cov})} = \beta d_k
\left\lceil
    C_2\left(
        d+\frac{\sigma^4s\log(d_1d_2)}{\varepsilon^2}
    \right)\right\rceil.
\]
The agent sends one fixed-length message consisting of the flag and the two payloads, after performing both local tests below. If either test fails, the flag is set to one and the payloads are padded with fixed bits. Otherwise the flag is zero and the payloads encode the two quantization indices.

\subsubsection*{Step 3: Self-Covariance Quantization~\label{sec:protocolstep3}}
\textbf{Empirical Estimation:} Agent $k$ locally computes its estimate of $C_{kk}$ as $$\tilde{C}_{kk} = \frac{1}{m}\sum_{i=1}^m X_k^{(i)}{X_k^{(i)}}^T.$$
\noindent \textbf{Quantization of Estimated Self-Covariance Matrices:} By Lemma~\ref{lem:f1crosscov-conc}, $\tilde{C}_{kk}$ satisfies $\|\tilde{C}_{kk}\|_{op}<11\sigma^2$ with high probability. If this holds, Agent $k$ quantizes $\tilde{C}_{kk}$ using the covering-net scheme in Section~\ref{sec:packing} for the operator-norm ball $\mathcal{B}^{d_k\times d_k}_{\|\cdot\|_{op}}(11\sigma^2)$ with $B_k^{(\text{self})}$ bits. The server decodes $\hat{C}_{kk}$ satisfying $\|\hat{C}_{kk}-\tilde{C}_{kk}\|_{op}\leq\omega'_k$, where $\omega'_k=33\sigma^2\cdot 2^{-B_k^{(\text{self})}/d_k^2}$. Before assembly, the server symmetrizes the decoded matrix. Since $\tilde C_{kk}$ is symmetric and $\| A^T\|_\op=\|A\|_\op$, write
\[
\bigg\|\frac{\hat C_{kk} + \hat{C}^T_{kk}}{2}-\tilde C_{kk}\bigg\|_\op
\le\frac{\|\hat C_{kk}-\tilde C_{kk}\|_\op + \|(\hat C_{kk}-\tilde C_{kk})^T\|_{\op}}{2}\le \omega'_k.
\]
Thus, the replacement $\hat C_{kk} \leftarrow \frac{\hat C_{kk}+\hat C_{kk}^T}{2}$ preserves the quantization guarantee and uses no additional bits. For notational simplicity, we continue to denote the resulting symmetric matrix by $\hat C_{kk}$. If $\|\tilde{C}_{kk}\|_{op}\geq11\sigma^2$, Agent $k$ sends an error signal.
% \begin{remark}[Note to Self] To make this more self-contained, perhaps also quote the important parts of Appendix A.6.1 into an earlier section.
% \end{remark}
\subsubsection*{Step 4a: Subsampling for Cross-Covariance}
To approximate the cross-covariance, select $n=\big\lfloor(\frac{B_1^{\text{(cov)}}}{d_1}\land\frac{B_2^{\text{(cov)}}}{d_2})/\beta\land m\big\rfloor$ paired samples. Define $\mathbf{X}_1\in\mathbb{R}^{d_1\times n}$ and $\mathbf{X}_2\in\mathbb{R}^{d_2\times n}$ by concatenating the first $n$ samples from each agent, such that $\mathbf{X}_k=[X_k^{(1)},\dots, X_k^{(n)}]$. The empirical estimator for $C_{21}$ is then $\tilde{C}_{21}=\frac{1}{n}\mathbf{X}_2\mathbf{X}_1^T$. Theorem~\ref{thm:sachievability} ensures $1\le n \le m$, as verified in Appendix~\ref{app:achievfull}.
\subsubsection*{Step 4b: Data Quantization for Cross-Covariance~\label{sec:protocolstep4}}
By Lemma~\ref{lem:f3subG-opnorm}, $\|\mathbf{X}_k\|_{op}<6\sigma\sqrt{d_k+n}$ with high probability. If this holds, Agent $k$ quantizes $\mathbf{X}_k$ using the covering-net scheme of Section~\ref{sec:packing} for the operator-norm ball $\mathcal{B}^{d_k\times n}_{\|\cdot\|_{op}}(6\sigma\sqrt{d_k+n})$ with $B_k^{(\text{cov})}$ bits. The server receives $\hat{\mathbf{X}}_k$ satisfying $\|\hat{\mathbf{X}}_k-\mathbf{X}_k\|_{op}\leq\omega''_k$, where $\omega''_k=18\sigma\sqrt{d_k+n}\cdot 2^{-B_k^{(\text{cov})}/(nd_k)}$.
% By Lemma~\ref{lem:f3subG-opnorm}, $\|\mathbf{X}_k\|_{op}\leq 6\sigma\sqrt{d_k+n}$ with high probability. If this holds, Agent $k$ quantizes $\mathbf{X}_k$ using the covering-net scheme of Section~\ref{sec:packing} for the operator-norm ball $\mathcal{B}^{nd_k}_{\|\cdot\|_{op}}(6\sigma\sqrt{d_k+n})$ with $B_k^{(\text{cov})}$ bits. The server receives $\hat{\mathbf{X}}_k$ satisfying \[
% \|\hat{\mathbf{X}}_k - \mathbf{X}_k\|_\op\le \omega''_k,\qquad \omega_k''=18\sigma\sqrt{d_k+n}\cdot 2^{-B_k^{(\text{cov})}/(nd_k)}.
% \]
If $\|\mathbf{X}_k\|_{op}\geq6\sigma\sqrt{d_k+n}$, Agent $k$ sends an error signal. If an error signal is received from either agent, output a zero matrix $\hat{C}=0$.

\subsubsection*{Step 5: Server Reconstruction of Covariance}
Upon receiving the quantized data $\hat{\mathbf{X}}_1,\hat{\mathbf{X}}_2$, the central server initially computes the estimate $\hat{C}_{21}=\frac{1}{n}\hat{\mathbf{X}}_2\hat{\mathbf{X}}_1^T$. If an error signal is received from either agent, the server outputs the zero matrix $\hat{C}=0$.

\paragraph{Thresholding.} For $i\in[d_2],j\in[d_1]$, set
$$[\hat{C}_{21}^{\text{threshold}}]_{ij}=\begin{cases}
    [\hat{C}_{21}]_{ij}\text{, if }|[\hat{C}_{21}]_{ij}|\geq\lambda^*,\\
    0,\text{ otherwise},
\end{cases}$$
where $\lambda^*:=\lambda+Q_{\text{max}}$, with $\lambda=c_0\sigma^2\sqrt{\log(d_1d_2)/n}$ where $c_0$ is a parameter to be determined in the proof, and $Q_{\text{max}}=1080\sigma^2\cdot 2^{-\beta}$.
Then assemble the block matrix
$$\hat{C}^*=\begin{pmatrix}
    \hat{C}_{11} & (\hat{C}_{21}^{\text{threshold}})^T\\
    \hat{C}_{21}^{\text{threshold}} & \hat{C}_{22}
\end{pmatrix}.$$
\paragraph{Spectral Decomposition.} The matrix $\hat C^*$ is symmetric by Step 3. The server performs a spectral decomposition and retains its non-negative eigenvalues, and finally returns
\begin{equation}~\label{eq:achievestimator}
    \hat{C}=\hat{C}^*_+.
\end{equation}
This produces a positive semi-definite estimate.

\begin{remark}
The proof of Theorem~\ref{thm:sachievability} bounds the distortion through a block decomposition into self-covariance and cross-covariance errors. The self-covariance analysis follows~\cite[Section~G.2]{rahmani2025fundamental}. The cross-covariance analysis decomposes the entry-wise error into quantization noise $Q_{ij}$ and statistical noise $S_{ij}$, controls these on the event $\mathcal{E}^c$, and separately sums over the non-zero and zero entries of $C_{21}$. The complete calculation is given in Appendix~\ref{app:achievfull}.
\end{remark}

\section{Discussion}~\label{sec:discussion}
% We have established near-optimal minimax bounds for distributed covariance matrix estimation in the vertical-split model under elementwise $s$-sparsity of the cross-covariance $C_{21}$.  The communication budget per agent drops from $\Omega(\sigma^4 d_k d_1 d_2/\varepsilon^2)$ in the dense setting (\cite{rahmani2025fundamental}) to $\Omega(\sigma^4 d_k s'\log(d_1 d_2/s')/\varepsilon^2)$ under $s$-sparsity, and a matching achievable scheme based on covering-net quantization and entry-wise hard thresholding confirms that this improvement is tight up to polylogarithmic factors.
We have established minimax lower bounds and an achievable scheme for distributed covariance matrix estimation in the vertical-split model under elementwise $s$-sparsity of the cross-covariance $C_{21}$. In their respective validity regimes, the communication lower bounds per agent scale as $\Omega(\sigma^4 d_k d_1 d_2/\varepsilon^2)$ in the dense setting (\cite{rahmani2025fundamental}) and $\Omega(\sigma^4 d_k s'\log(d_1 d_2/s')/\varepsilon^2)$ under $s$-sparsity. Under the conditions of Remark~\ref{rem:achievmatch}, our achievable scheme based on covering-net quantization and entry-wise hard thresholding matches the sparse lower-bound rate up to polylogarithmic factors for its cross-covariance communication component.

\paragraph{Open problems.}
We present a few possible extensions of the work.

\emph{Unknown sparsity level.}
Our achievable scheme requires knowledge of $s$ to calibrate the threshold $\lambda^*$.  In practice, $s$ is rarely known.  Whether a protocol that requires no prior knowledge of $s$ can achieve the same $s\log(d_1 d_2/s)$ rate in the vertical split remains open.

\emph{Other structural assumptions.}
Beyond elementwise sparsity, questions on whether other structural assumptions on covariance matrices such as bandedness, low-rank structure, Toeplitz, and group sparsity similarly reduce communication cost are natural extensions of the present work.

\bibliographystyle{plainnat}
\bibliography{references}

\appendix
\section{Verification of Hypothesis Family Properties}
\label{app:verification}

For the $s$-sparse hypothesis family, we have:

\subsection*{$C_v^{(w)}$ is positive definite,
$\forall v \in \mathcal{V}, w \in \mathcal{W}$.}

\begin{proof}
    By the Schur complement criterion
    (Lemma~\ref{lem:schur}), $C_v^{(w)} \succ 0$ if and
    only if
    \[
        I_{d_2}
        -
        \delta^2
        E_v^{(w)}(E_v^{(w)})^T
        \succ 0.
    \]
    Since $E_v^{(w)}$ is a signed partial permutation,
    \[
        E_v^{(w)}(E_v^{(w)})^T
        =
        \sum_k e_{j_k}e_{j_k}^T
    \]
    has eigenvalues in $\{0,1\}$. Hence
    \[
        I_{d_2}
        -
        \delta^2
        E_v^{(w)}(E_v^{(w)})^T
    \]
    has minimum eigenvalue $1-\delta^2$.
    Since $\delta^2\leq 1/2$, we have
    \[
        1-\delta^2\geq \frac12>0,
    \]
    and therefore $C_v^{(w)}\succ0$.
\end{proof}

\subsection*{The sub-Gaussianity condition is satisfied, i.e.,
$P_v^{(w)} \in \mathrm{subG}^{(d)}(\sigma)$,
$\forall v \in \mathcal{V}, w \in \mathcal{W}$.}

\begin{proof}
   Utilizing Definition~\ref{def:subgaussianvec} and Remark~3.3,
   since $\|E_v^{(w)}\|_{\op}=1$ and
   $\delta\leq 1/\sqrt{2}<1$, we have
    \[
        \|C_v^{(w)}\|_{\op}
        =
        \frac{\sigma^2}{2}(1+\delta)
        \leq \sigma^2,
    \]
    and we are done.
\end{proof}

\subsection*{Separation is characterized by
$\rho_F=\sqrt{2}\rho_F^{\mathrm{(cross)}}$.}

\begin{proof}
Since the diagonal blocks
\[
    C_{11}=\frac{\sigma^2}{2}I_{d_1}
    \qquad\text{and}\qquad
    C_{22}=\frac{\sigma^2}{2}I_{d_2}
\]
are identical across all hypotheses, the difference $C_v^{(w)}-C_{v'}^{(w)}$ is zero on the diagonal blocks. Its lower-left block is $\frac{\sigma^2\delta}{2}\big(E_v^{(w)}-E_{v'}^{(w)}\big)$, and its upper-right block is the transpose of this block.
% are identical across all hypotheses, the difference
% $C_v^{(w)}-C_{v'}^{(w)}$ is zero on the diagonal blocks and
% equals
% \[
%     \pm\frac{\sigma^2\delta}{2}
%     \big(E_v^{(w)}-E_{v'}^{(w)}\big)
% \]
% on the two off-diagonal blocks. 
Recall that, for
$w=(w_L,w_R)\in\mathcal W$,
\[
    E_v^{(w)}=w_LE_vw_R.
\]
Hence
\[
    E_v^{(w)}-E_{v'}^{(w)}
    =
    w_L(E_v-E_{v'})w_R.
\]
Since $w_L$ and $w_R$ are signed permutation matrices, they are
orthogonal, and therefore
\[
    \|E_v^{(w)}-E_{v'}^{(w)}\|_F
    =
    \|E_v-E_{v'}\|_F.
\]
In particular, the separation is independent of $w$.

Let
\[
    D_v^{(w)}
    :=
    (C_v^{(w)})_{21}
    =
    \frac{\sigma^2\delta}{2}E_v^{(w)}.
\]
Then
\[
    \big\|C_v^{(w)}-C_{v'}^{(w)}\big\|_F^2
    =
    \big\|(D_v^{(w)}-D_{v'}^{(w)})^T\big\|_F^2
    +
    \big\|D_v^{(w)}-D_{v'}^{(w)}\big\|_F^2
    =
    2\big\|D_v^{(w)}-D_{v'}^{(w)}\big\|_F^2,
\]
where the first equality uses $C_{12}=C_{21}^T$ and the
second uses $\|A^T\|_F=\|A\|_F$. Taking the infimum over
$w\in\mathcal W$ and $v\neq v'$ on both sides and
square-rooting gives
\[
    \rho_F
    =
    \sqrt{2}\,
    \rho_F^{\mathrm{(cross)}}.
\]
\end{proof}

\section{Self-Covariance and Full Covariance Lower Bounds~\label{app:lowerbounds}}
\subsection{Sample Complexity Lower Bound for $C_{21}$~\label{app:crosssample}}
We utilise the same hypothesis family
Eq.~\eqref{eq:ssparse-Cvw}, with the omission of $W$. More precisely, consider the same index set $\mathcal{V}=[1:|\tilde{\mathcal{G}}|]$ and a corresponding family of distributions $\mathcal{P}_{\mathcal{V}}=\{P_v\}_{v\in\mathcal{V}}$, where $P_v=N(0,C_v)$ with
\begin{equation}
    C_v=\frac{\sigma^2}{2}\begin{pmatrix}
        I_{d_1}&\delta E_v^T\\\delta E_v&I_{d_2}
    \end{pmatrix},
\end{equation}
where $E_v$ is defined in Eq.~\ref{eq:base-Ev}, $E_v\in\R^{d_2\times d_1}$ with $\|E_v\|_{\text{op}}\le 1$, and $\delta^2\le1/2$ is again a parameter to be determined subsequently. By the chain-rule argument in Eq.~\ref{eq:chainrule}, with $W$ omitted and using the fact that the marginals do not depend on $V$, we have $I(V;M_1,M_2)\le I(M_1;M_2|V)$. By data processing under the conditional law given $V$, $I(M_1;M_2|V)\le I(\bX_1;\bX_2|V)$ because $M_1\to \mathbf{X_1}\to \mathbf{X_2}\to M_2$. Hence, we have $I(V;M_1,M_2)\le I(\mathbf{X}_1;\mathbf{X}_2|V)$.

\noindent Given $m$ i.i.d.\ samples, the mutual information tensorizes:
$I(\mathbf{X}_1;\mathbf{X}_2|V=v)
= m\,I(X_1;X_2|V=v)$. By Lemma~\ref{lem:schur} with $E_v E_v^T = \sum_{k=1}^{s'}e_{j_k}e_{j_k}^T$ having $s'$ eigenvalues equal to 1, we have
\begin{align}~\label{eq:detcv}
    \det(C_v) &= \left(\frac{\sigma^2}{2}\right)^{d_1+d_2}
    \det(I_{d_2}-\delta^2 E_vE_v^T)
    = \left(\frac{\sigma^2}{2}\right)^{d_1+d_2}
    (1-\delta^2)^{s'}.
\end{align}

\noindent Then, we have
\begin{align}~\label{eq:ssparsesamplecrossinfo}
    I(\mathbf{X_1};\mathbf{X_2}|V=v)&=mI(X_1;X_2|V=v)\nonumber\\
    &=m[h(X_1|V=v)+h(X_2|V=v)-h(X_1,X_2|V=v)]\nonumber\\
    &=\frac{m}{2}\log\bigg(\frac{\det(\frac{\sigma^2}{2}I_{d_1})\det(\frac{\sigma^2}{2}I_{d_2})}{\det(C_v)}\bigg)\nonumber\\
    &=-\frac{s'm}{2}\log(1-\delta^2)\nonumber\\
    &\le s'm\delta^2.
\end{align}
The last inequality follows since $\delta^2\le\frac{1}{2}$, and the third equality follows by Theorem~\ref{thm:gaussianentropy}.

Substituting Eqs.~\ref{eq:ssparse-packing-number},~\ref{eq:ssparse-rho-value}, and~\ref{eq:ssparsesamplecrossinfo} into Eq.~\ref{eq:fanodccme}, we have
\begin{equation}\label{eq:ssparse-sample-fano}
    \MF^{\text{cross}}\ge\frac{\sigma^2\delta\sqrt{s'}}{4\sqrt{2}}\bigg(1-\frac{s'm\delta^2+1}{s'L}\bigg).
\end{equation}
Choose
\begin{equation}\label{eq:ssparse-sample-delta}
    \delta^2 = \frac{L}{2m}\;\wedge\;\frac{1}{2}.
\end{equation}
We then have the following bound:
\begin{proposition}[$s$-Sparse Cross-Covariance Sample Bound]
\label{prop:ssparse-sample-cross}
Assume $s'L\ge 4$, where
$L=\frac{3}{4}\log(2d_1d_2)-\log s'-\log(4e)$.
Then,
\begin{equation}\label{eq:ssparse-sample-final}
    \MFcross\geq\frac{\sigma^2}{32}
    \bigg(\sqrt{\frac{s'L}{m}}\;\wedge\;
    \sqrt{s'}\bigg).
\end{equation}
In particular, whenever $d_1d_2/s'$ is sufficiently large so that $L=\Omega\big(\log\frac{d_1d_2}{s'}\big)$, we have
\begin{equation}
    \MFcross=\Omega\bigg(\sigma^2\left[
    \sqrt{\frac{s'\log(d_1d_2/s')}{m}}\land \sqrt{s'}\right]\bigg).
\end{equation}
\end{proposition}

\subsection{Sample Complexity Lower Bound for Self-Covariance~\label{app:fullsample}}
The bound contains a $\frac{d_1\lor d_2}{\sqrt{m}}$ term. This is because, since $\|\hat{C}-C\|_F\geq\|\hat{C}_{22}-C_{22}\|_F$, it suffices to lower-bound the minimax risk of estimating $C_{22}$ alone.  The class $\mathcal{P}_s$ contains all $\sigma$-sub-Gaussian distributions with $C_{21}=0$ and arbitrary $C_{22}$ satisfying $\|C_{22}\|_{\mathrm{op}}\leq\sigma^2$; over this subclass, the sparsity constraint is trivially satisfied for every $s$.  The centralized minimax risk for estimating a $p\times p$ covariance matrix from $m$ i.i.d.\ samples under the Frobenius norm is $\Omega(\sigma^2(p/\sqrt{m}\wedge\sqrt{p}))$ (\cite{ABD20,DMR20}), and by monotonicity this extends to the distributed setting. Applying this with $p=d_2$, and by symmetry with $p=d_1$, yields the result.

For completeness, we describe an adaptation of the argument of~\cite[Section~E.1.2]{rahmani2025fundamental} by splitting within Agent~2's coordinates and perturbing the intra-agent cross-covariance $C_{2a,2b}$ with dense matrices while setting $C_{21}=0$, so that the sparsity constraint is trivially satisfied and the estimation difficulty arises from the dense self-covariance.

\begin{proposition}[Full Covariance Sample Bound]
\label{prop:sample-full}
Suppose $d_1\vee d_2\geq 4$.  Then for any $s$ and any
DCME scheme with parameters $(\sigma,m,d_{1:2},B_{1:2})$,
\begin{equation}\label{eq:sample-full}
    \MF \;\geq\; \frac{\sigma^2}{7168}
    \left(\frac{d_1\vee d_2}{\sqrt{m}}
    \;\wedge\;\sqrt{2(d_1\vee d_2)}\right).
\end{equation}
In particular,
\begin{equation}\label{eq:sample-full-asymp}
    \MF \;=\; \Omega\!\left(\sigma^2
    \left(\frac{d_1\vee d_2}{\sqrt{m}}
    \;\wedge\;\sqrt{d_1\vee d_2}\right)\right).
\end{equation}
\end{proposition}
\begin{proof}
Without loss of generality assume $d_2\geq d_1$.  Define
$X_{2a}=(X_2)_{[1:d_2']}$ and
$X_{2b}=(X_2)_{[d_2'+1:d_2]}$ with
$d_2'=\lfloor d_2/2\rfloor$, $d_2''=\lceil d_2/2\rceil$.
For each $v\in\mathcal{V}$, define
\[
    C_v=\frac{\sigma^2}{2}\begin{pmatrix}
    I_{d_1} & 0 & 0\\
    0 & I_{d_2'} & \delta F_v^\top\\
    0 & \delta F_v & I_{d_2''}
    \end{pmatrix},
\]
where $\{F_v\}_{v\in\mathcal{V}}$ is an $\omega$-packing under the Frobenius norm of the operator-norm unit ball in $\mathbb{R}^{d_2''\times d_2'}$, with $\omega=\sqrt{d_2'}/56$ and $|\mathcal{V}|\geq4^{d_2'd_2''}$, whose existence follows from~\cite[Lemma~A.10]{rahmani2025fundamental}. Since $(C_v)_{21}=0$ for all $v$, the sparsity constraint is satisfied for every $s$.

The remainder of the proof follows the argument of~\cite[Section~E.1.2]{rahmani2025fundamental} with the substitution $(d/2,\,d/2)\to(d_2',\,d_2'')$. For brevity, we refer the reader to \cite{rahmani2025fundamental}. The packing cardinality satisfies $\log|\mathcal{V}|\geq d_2^2/4$, the separation satisfies $\rho_F\geq\delta\sigma^2\sqrt{d_2}/448$, and the mutual information satisfies $I(V;M_1,M_2)\le I(V;\bX_{2a},\bX_{2b})\leq I(\bX_{2a};\bX_{2b}\mid V) \leq md_2\delta^2/2$.  Substituting into Fano's inequality and choosing $\delta^2=(d_2/(4m))\wedge\tfrac{1}{2}$ yields
\[
    \MF\geq\frac{\sigma^2}{7168}\left(
    \frac{d_2}{\sqrt{m}}\;\wedge\;\sqrt{2d_2}\right).
\]
By symmetry (splitting Agent~1 if $d_1>d_2$), $\MF\geq\frac{\sigma^2}{7168} ((d_1\vee d_2)/\sqrt{m}\;\wedge\;\sqrt{2(d_1\vee d_2)})$.
\end{proof}
% \begin{proof}
% Without loss of generality assume $d_2\geq d_1$.  Define
% $X_{2a}=(X_2)_{[1:d_2']}$ and
% $X_{2b}=(X_2)_{[d_2'+1:d_2]}$ with
% $d_2'=\lfloor d_2/2\rfloor$, $d_2''=\lceil d_2/2\rceil$.
% For each $v\in\mathcal{V}$, define
% \[
%     C_v=\frac{\sigma^2}{2}\begin{pmatrix}
%     I_{d_1} & 0 & 0\\
%     0 & I_{d_2'} & \delta F_v^\top\\
%     0 & \delta F_v & I_{d_2''}
%     \end{pmatrix},
% \]
% where $\{F_v\}_{v\in\mathcal{V}}$ is an $\omega$-packing of the operator-norm unit ball in $\mathbb{R}^{d_2''\times d_2'}$ from~\cite[Appendix~A.6]{rahmani2025fundamental}. Since $(C_v)_{21}=0$ for all $v$, the sparsity constraint is satisfied for every $s$.

% The remainder of the proof is identical to~\cite[Section~E.1.2]{rahmani2025fundamental} with the substitution $(d/2,\,d/2)\to(d_2',\,d_2'')$. For brevity, we refer the reader to \cite{rahmani2025fundamental}. The packing number satisfies $\log|\mathcal{V}|\geq d_2^2/4$, the separation satisfies $\rho_F\geq\delta\sigma^2\sqrt{d_2}/448$, and the mutual information satisfies $I(V;M_1,M_2)\le I(V;\bX_{2a},\bX_{2b})\leq I(\bX_{2a};\bX_{2b}\mid V) \leq md_2\delta^2/2$.  Substituting into Fano's inequality and choosing $\delta^2=(d_2/(4m))\wedge\tfrac{1}{2}$ yields
% \[
%     \MF\geq\frac{\sigma^2}{7168}\left(
%     \frac{d_2}{\sqrt{m}}\;\wedge\;\sqrt{2d_2}\right).
% \]
% By symmetry (splitting Agent~1 if $d_1>d_2$), $\MF\geq\frac{\sigma^2}{7168} ((d_1\vee d_2)/\sqrt{m}\;\wedge\;\sqrt{2(d_1\vee d_2)})$.
% \end{proof}

\subsection{Communication Lower Bound for Self-Covariance~\label{app:self-comm}}
We utilize the result established in~\cite[Appendix~E.2]{rahmani2025fundamental} via a hypothesis family that varies $C_{kk}$ while setting $C_{21}=0$.  Since $|C_{21}|_0=0\leq s$ for every $s$, the family lies in $\mathcal{P}_s$ and the bound applies unchanged.
\begin{proposition}[Self-Covariance Communication Bound]
\label{prop:comm-self}
\textup{(\cite[Appendix~E.2, Eq.~(148)]{rahmani2025fundamental}.)}
For any $s$ and any DCME scheme with parameters
$(\sigma,m,d_{1:2},B_{1:2})$,
\begin{equation}\label{eq:comm-self}
    \MF \;\geq\; \frac{\sigma^2}{56}\left(
    \sqrt{d_1}\cdot 2^{-16B_1/d_1^2}
    \;\vee\;
    \sqrt{d_2}\cdot 2^{-16B_2/d_2^2}\right).
\end{equation}
This bound is independent of the sparsity parameter $s$. In particular,
\begin{equation}\label{eq:comm-self-asymp}
    \MF =\Omega\bigg(\sigma^2\left(
    \sqrt{d_1}\cdot 2^{-16B_1/d_1^2}
    \;\vee\;
    \sqrt{d_2}\cdot 2^{-16B_2/d_2^2}\right)\bigg).
\end{equation}
\end{proposition}

% \begin{remark}\label{rem:comm-self}
% The bound~\eqref{eq:comm-self} decays exponentially in
% $B_k/d_k^2$: once $B_k\gg d_k^2$, i.e., there are enough
% bits to encode the $\sim d_k^2$ free parameters of the
% symmetric matrix $C_{kk}$, this term becomes negligible
% and the cross-covariance communication bound
% (Proposition~\ref{prop:ssparse-comm}) dominates.
% \end{remark}

\section{Proof of Achievability Theorem~\label{app:achievfull}}
The mean-removal, covering-net and self-covariance arguments follow \cite[Section~G.2]{rahmani2025fundamental}. The treatment of the thresholded cross-covariance matrix is given in detail below. We prove the existence of a scheme with expected Frobenius distortion less than $\varepsilon$.

\noindent As specified in Step 1 of Section~\ref{sec:protocolstep1}, $m$ now denotes the number of mean-zero observations after paired differencing. Thus the assumptions of Theorem~\ref{thm:sachievability} give:
\begin{equation}\label{eq:sachibudgetbound}
    m\geq\max\left\{ 
    \bigg\lceil2^{19}\frac{d}{\tilde{\varepsilon}^2}\bigg\rceil,
    \bigg\lceil C_2\bigg(d+\frac{\sigma^4 s\log(d_1d_2)}{\varepsilon^2}\bigg)\bigg\rceil
    \right\},
\end{equation}
\begin{equation}\label{eq:sachicommbound}
    B_k\geq 1+\underbrace{\beta d_k\bigg\lceil C_2\bigg(d+\frac{\sigma^4s\log(d_1d_2)}{\varepsilon^2}\bigg)\bigg\rceil}_{\text{cross-covariance budget}}
    + \underbrace{\bigg\lceil2d_k^2\log\!\left(\frac{528}{\tilde{\varepsilon}}\right)\bigg\rceil}_{\text{self-covariance budget}},
\end{equation}
where \begin{equation}\label{eq:sachibeta}
    \beta=\max\bigg\{2,\bigg\lceil\log\bigg(\frac{C_3\sigma^2\sqrt{s}}{\varepsilon}\bigg)\bigg\rceil\bigg\},
\end{equation}
and $C_2,C_3$ are 
universal constants to be chosen later. Also let 
\begin{equation}\label{eq:sachisubsetbound}
    n=\left\lfloor\frac{B_1^{(\mathrm{cov})}/d_1\;\wedge\; 
    B_2^{(\mathrm{cov})}/d_2}{\beta}\;\wedge\; m\right\rfloor.
\end{equation}
By Eqs.~\ref{eq:sachibudgetbound}, \ref{eq:sachicommbound}, \ref{eq:sachisubsetbound} and the allocation in Step 2, we have
\begin{equation}\label{eq:sachisubsetbound2}
    n\geq C_2\bigg(d+\frac{\sigma^4 s\log(d_1d_2)}{\varepsilon^2}\bigg).
\end{equation}
In particular, $1\le n\le m$ when $C_2\ge 1$.

\noindent From Steps 3--4 of the protocol, we define the following error events naturally:
$$\mathcal{E}_{k,1}:=\{\|\tilde{C}_{kk}\|_{op}\geq 11\sigma^2\}\hspace{10mm}\mathcal{E}_{k,2}:=\{\|\mathbf{X}_k\|_{op}\geq 6\sigma\sqrt{d_k+n}\}.$$
These are the events on which the corresponding norm test fails and the agent sends an error signal. Let
\begin{align*}
\mathcal{E} := \mathcal{E}_{1,1}\cup\mathcal{E}_{2,1}\cup\mathcal{E}_{1,2}\cup\mathcal{E}_{2,2}.
\end{align*}
By Lemma~\ref{lem:f1crosscov-conc}, we have
$$\mathbb{P}(\mathcal{E}_{k,1})\leq\min\{1,\exp(6d_k-m)\},$$
and by Lemma~\ref{lem:f3subG-opnorm}, we have
$$\mathbb{P}(\mathcal{E}_{k,2})\le\exp(-2(d_k+n)).$$
Taking a union bound, we have
\begin{align}
    \mathbb{P}[\mathcal{E}]&\leq \mathbb{P}[\mathcal{E}_{1,1}]+\mathbb{P}[\mathcal{E}_{2,1}]+\mathbb{P}[\mathcal{E}_{1,2}]+\mathbb{P}[\mathcal{E}_{2,2}]\nonumber\\&\leq\exp(6d_1-m)+\exp(6d_2-m)+\exp(-2(d_1+n))+\exp(-2(d_2+n))\nonumber\\&\leq 2\exp(6d-m)+2\exp(-2(n+1))\nonumber\\
    &<\frac{\tilde{\varepsilon}}{10000},\label{eq:sachieveunionrrorbound}
\end{align}
where the last inequality follows from the following two inequalities:
\begin{enumerate}[label=(\roman*)]
   \item By Eq.~\ref{eq:sachibudgetbound} and $0<\tilde\varepsilon\le 1$, we have
   \[
   6d-m\le d\bigg(6-\frac{2^{19}}{\tilde\varepsilon^2}\bigg)\le-\frac{2^{18}d}{\tilde\varepsilon^2}\le-\frac{2^{18}}{\tilde\varepsilon^2}.
   \]
   Using $e^{-x}\le 1/x$ for $x>0$,
   \[
   2\exp(6d-m)\le 2\exp(-2^{18}/\tilde\varepsilon^2)\le\frac{\tilde\varepsilon^2}{2^{17}}\le\frac{\tilde\varepsilon}{2^{17}}<\frac{\tilde\varepsilon}{20000}.
   \]

   \item Since $\varepsilon^2=\sigma^4 d\tilde\varepsilon^2$, $s\ge 1$ and $\log(d_1d_2)\ge 1$, Eq.~\ref{eq:sachisubsetbound2} implies
   \[
   n\ge  C_2\bigg(d+\frac{s\log(d_1d_2)}{d\tilde\varepsilon^2}\bigg)\ge\frac{2C_2\sqrt{s\log(d_1d_2)}}{\tilde\varepsilon}\ge\frac{2C_2}{\tilde{\varepsilon}}.
   \]
   Consequently, for $C_2\ge10000$, 
   \[ 2\exp(-2(n+1))\le\frac{1}{n+1}\le\frac1n\le\frac{\tilde\varepsilon}{2C_2} \le\frac{\tilde\varepsilon}{20000}.\]
\end{enumerate}
\noindent In particular, we have
\begin{equation}
    \mathbb{P}(\mathcal{E}^c)\geq0.9999.
\end{equation}
By the Law of Total Expectation, decompose
$$\mathbb{E}[\|\hat{C}-C\|_F] = \mathbb{E}[\|\hat{C}-C\|_F|\mathcal{E}]\mathbb{P}[\mathcal{E}] + \mathbb{E}[\|\hat{C}-C\|_F|\mathcal{E}^c]\mathbb{P}[\mathcal{E}^c].$$
% \noindent Our goal is to show the above $<\varepsilon$. Now, we analyze the above term by term.\\
\noindent On $\mathcal{E}$, $\hat{C}=0$. Hence,
    $$\mathbb{E}[\|\hat{C}-C\|_F|\mathcal{E}]\mathbb{P}[\mathcal{E}]=\|C\|_F \mathbb{P}[\mathcal{E}] \leq \sigma^2\sqrt{d}\cdot \frac{\tilde{\varepsilon}}{10000}=\frac{\varepsilon}{10000}.$$
Also,
$$\mathbb{E}[\|\hat{C}-C\|_F|\mathcal{E}^c]\mathbb{P}[\mathcal{E}^c]\leq \mathbb{E}[\|\hat{C}-C\|_F|\mathcal{E}^c].$$
It suffices to show $$\mathbb{E}[\|\hat{C}-C\|_F|\mathcal{E}^c]<\varepsilon-\frac{\varepsilon}{10000}.$$
\noindent First, we remark that taking the positive semidefinite part does not increase the Frobenius error. By Step 3 in Section~\ref{sec:protocolstep3}, $\hat C^*$ is symmetric, and hence
$$\mathbb{E}[\|\hat{C}-C\|_F|\mathcal{E}^c]= \mathbb{E}[\|\hat{C}^*_+-C\|_F|\mathcal{E}^c]\leq \mathbb{E}[\|\hat{C}^*-C\|_F|\mathcal{E}^c],$$
where the equality follows by construction, and the inequality follows by:
\begin{align*}
\|\hat{C}^* - C\|_F^2 
&= \operatorname{Tr}\left\{\left(\hat{C}^* - C\right)^2\right\} \\
&= \operatorname{Tr}\left\{\left(\hat{C}^* - \hat{C}_+^*\right)^2\right\} 
 + \operatorname{Tr}\left\{\left(\hat{C}_+^* - C\right)^2\right\} 
 + 2\operatorname{Tr}\left\{\left(\hat{C}^* - \hat{C}_+^*\right)\left(\hat{C}_+^* - C\right)\right\} \\
&\geq \|\hat{C}_+^* - C\|_F^2 
 + 2\operatorname{Tr}\left\{\left(\hat{C}^* - \hat{C}_+^*\right)\left(\hat{C}_+^* - C\right)\right\} \\
&= \|\hat{C}_+^* - C\|_F^2 
 + 2\operatorname{Tr}\left\{\left(\hat{C}_+^* - \hat{C}^*\right)C\right\} \\
&\geq \|\hat{C}_+^* - C\|_F^2.
\end{align*}
The last equality is true since the positive part $\hat{C}^*_+$ of $\hat{C}^*$ is orthogonal to the negative part of it, and the last inequality is true since the trace of the product of two positive semi-definite matrices is non-negative.\\
Hence, in the remainder of the proof, we focus on upper bounding $\mathbb{E}[\|\hat{C}^*-C\|_F|\mathcal{E}^c]$.\\
Utilize the elementary results: (1) For non-negative reals,
$\sqrt{a^2+b^2+c^2}\leq a+b+c$,
(2) For any $A\in\mathbb{R}^{p\times q}$, $\|A\|_F\leq\sqrt{\min(p,q)}\,\|A\|_{\mathrm{op}}$ to perform block decomposition:
\begin{align*}
    \mathbb{E}[\|\hat{C}^*-C\|_F|\mathcal{E}^c]
    &= \mathbb{E}\!\left[\left\|
    \begin{bmatrix}
        \hat{C}_{11}-C_{11}
        & (\hat{C}_{21}^{\,\mathrm{threshold}})^\top - C_{21}^\top\\
        \hat{C}_{21}^{\,\mathrm{threshold}}-C_{21}
        & \hat{C}_{22}-C_{22}
    \end{bmatrix}
    \right\|_F\bigg|\mathcal{E}^c\right]\\
    &\leq \mathbb{E}\big[\|\hat{C}_{11}-C_{11}\|_F
    \big|\mathcal{E}^c\big]
    + \mathbb{E}\big[\|\hat{C}_{22}-C_{22}\|_F
    \big|\mathcal{E}^c\big]
    + \sqrt{2}\,\mathbb{E}\big[\|\hat{C}_{21}^{\,\mathrm{threshold}}
    -C_{21}\|_F\big|\mathcal{E}^c\big]\\
    &\leq \underbrace{\sqrt{d_1}\,\mathbb{E}\big[\|\hat{C}_{11}
    -C_{11}\|_{\mathrm{op}}\big|\mathcal{E}^c\big]
    + \sqrt{d_2}\,\mathbb{E}\big[\|\hat{C}_{22}
    -C_{22}\|_{\mathrm{op}}\big|\mathcal{E}^c\big]}_{\text{Self-Covariance Error}}\\
    &\quad+ \underbrace{\sqrt{2}\,\mathbb{E}\big[
    \|\hat{C}_{21}^{\,\mathrm{threshold}}-C_{21}\|_F
    \big|\mathcal{E}^c\big]}_{\text{Cross-Covariance Error}}.
\end{align*}

\subsection*{Self-Covariance Error Analysis}
For each matrix to be quantized, we apply the covering-net scheme of Section~\ref{sec:packing} with the formula $\omega = 3r \cdot 2^{-B/(pq)}$, where $r$ is the ball radius, $B$ is the bit budget, and $p \times q$ is the matrix dimension.
\begin{enumerate}[label=(\roman*)]
\item Quantization of $\widetilde{C}_{kk} \in \mathbb{R}^{d_k \times d_k}$: $r = 11\sigma^2$, $B = B_k^{(\mathrm{self})}$, $pq = d_k^2$:
\begin{equation}\label{eq:omega_self}
\omega_k' = 33\sigma^2 \cdot 2^{-B_k^{(\mathrm{self})}/d_k^2}, \qquad k = 1, 2.
\end{equation}
Note that by Eq.~\ref{eq:sachicommbound}, we have $B_k^{(\text{self})}\ge2d_k^2\log(528/\tilde{\varepsilon})$, hence $\frac{B_k^{(\text{self})}}{d_k^2}
\ge2\log(\frac{528}{\tilde{\varepsilon}})$, and we have
$$\omega_k'\leq33\sigma^2\cdot2^{-2\log(528/\tilde{\varepsilon})}=33\sigma^2\cdot(\frac{\tilde{\varepsilon}}{528})^2\leq33\sigma^2\cdot\frac{\tilde{\varepsilon}}{528}=\frac{\sigma^2\tilde{\varepsilon}}{16}.$$

\item Quantization of $\mathbf{X}_k \in \mathbb{R}^{d_k \times n}$: $r = 6\sigma\sqrt{d_k + n}$, $B = B_k^{(\mathrm{cov})}$, $pq = nd_k$:
\begin{equation}\label{eq:omega_data}
\omega_k'' = 18\sigma\sqrt{d_k + n} \cdot 2^{-B_k^{(\mathrm{cov})}/(nd_k)}.
\end{equation}
Note that $n\le\frac{(\frac{B_1^{\text{(cov)}}}{d_1}\land\frac{B_2^{\text{(cov)}}}{d_2})}{\beta}\implies \frac{B_k^{\text{(cov)}}}{nd_k}\geq\beta$, hence 
$$\omega_k''\leq18\sigma\sqrt{d_k+n}\cdot2^{-\beta}.$$
\end{enumerate}
\noindent On $\mathcal{E}^c$, these guarantee 
\begin{equation}~\label{eq:achievopbound1}
    \|\hat{C}_{kk} - \widetilde{C}_{kk}\|_{\mathrm{op}} \leq \omega_k',
\end{equation}
\noindent and 
\begin{equation}~\label{eq:achievopbound2}
    \|\hat{\mathbf{X}}_k - \mathbf{X}_k\|_{\mathrm{op}} \leq \omega_k''.
\end{equation}
\noindent By definition of $\mathcal{E}^c_{k,2}$, we also have
\begin{equation}~\label{eq:achievopbound3}
    \|\mathbf{X}_k\|_{op}\leq 6\sigma\sqrt{d_k+n}.
\end{equation}
\noindent Moreover, by Proposition~\ref{prop:f2crosscov-expected} and $m\geq2^{19}d/\tilde{\varepsilon}^2\geq2^{19}d_k\implies\sqrt{2d_k/m}>2d_k/m$,
\begin{equation}~\label{eq:selfcovf2}
    \mathbb{E}[\|\tilde{C}_{kk}-C_{kk}\|_{op}]\leq32\sigma^2\max\left\{\sqrt{\frac{2d_k}{m}},\frac{2d_k}{m}\right\}\leq 32\sigma^2\sqrt{\frac{2d_k}{m}}.
\end{equation}
Conditioning on $\mathcal{E}^c$ gives
\begin{equation}~\label{eq:selfcovcond}
    \mathbb{E}[\|\tilde{C}_{kk}-C_{kk}\|_{op}|\mathcal{E}^c]\leq\frac{32\sigma^2\sqrt{2d_k/m}}{\mathbb{P}[\mathcal{E}^c]}\leq33\sigma^2\sqrt{\frac{2d_k}{m}},
\end{equation}
with the last inequality following from $\mathbb{P}[\mathcal{E}^c]\geq0.9999$.\\
\noindent Again, by Eq.~\ref{eq:sachibudgetbound}, we substitute $m\geq2^{19}d/\tilde{\varepsilon}^2$ and $d_k\leq d$ and get
\begin{equation}
    33\sigma^2\sqrt{\frac{2d_k}{m}}\leq33\sigma^2\sqrt{\frac{2d\tilde{\varepsilon}^2}{2^{19}d}}=33\sigma^2\sqrt{\frac{2\tilde{\varepsilon}^2}{2^{19}}}=\frac{33\sigma^2\tilde{\varepsilon}}{512}
\end{equation}
\noindent Hence,
\begin{align*}
    \mathbb{E}[\|\hat{C}_{kk}-C_{kk}\|_{op}|\mathcal{E}^c] &\leq \mathbb{E}[\|\hat{C}_{kk}-\tilde{C}_{kk}\|_{op}|\mathcal{E}^c]+\mathbb{E}[\|\tilde{C}_{kk}-C_{kk}\|_{op}|\mathcal{E}^c]\\
    &\leq \frac{\sigma^2\tilde{\varepsilon}}{16}+\frac{33\sigma^2\tilde{\varepsilon}}{512}\leq\frac{\sigma^2\tilde{\varepsilon}}{7}.
\end{align*}
Thus, it is clear that
\begin{equation}~\label{eq:selfcoverror}
    \sqrt{d_k}\mathbb{E}[\|\hat{C}_{kk}-C_{kk}\|_{op}|\mathcal{E}^c]\leq\frac{\varepsilon}{7}.
\end{equation}

\subsubsection*{Cross-Covariance Error Analysis}
We first bound $\mathbb{E}[\|\hat{C}_{21}^{\text{threshold}}-C_{21}\|^2_F|\mathcal{E}^c]$ element-wise, then take square roots using Jensen's inequality ($\mathbb{E}[\|A\|_F]\leq\sqrt{\mathbb{E}[\|A\|^2_F]}$).

\noindent \paragraph{Decomposing the threshold error.} For each entry $(i,j)$ with $i\in[d_2]$, $j\in[d_1]$,
the pre-thresholding error decomposes as
\begin{equation}\label{eq:achievprethresholdingdecomp}
    [\hat{C}_{21}]_{ij}-[C_{21}]_{ij}
    =\underbrace{[\hat{C}_{21}]_{ij}
    -[\widetilde{C}_{21}]_{ij}}_{Q_{ij}}
    +\underbrace{[\widetilde{C}_{21}]_{ij}
    -[C_{21}]_{ij}}_{S_{ij}},
\end{equation}
where we call $Q_{ij}$ the \emph{quantisation noise} and
$S_{ij}$ the \emph{statistical noise}, given by
$$
    Q_{ij}
    =\bigg[\frac{1}{n}\hat{\bX}_2\hat{\bX}_1^\top
    \bigg]_{ij}
    -\bigg[\frac{1}{n}\bX_2\bX_1^\top\bigg]_{ij},
    \qquad
    S_{ij}
    =\frac{1}{n}\sum_{l=1}^n
    \big([X_2^{(l)}]_i[X_1^{(l)}]_j-[C_{21}]_{ij}\big).
$$
By the triangle inequality,
\begin{equation}\label{eq:achievprethresholdingdecompabs}
    \big|[\hat{C}_{21}]_{ij}-[C_{21}]_{ij}\big|
    \leq|Q_{ij}|+|S_{ij}|.
\end{equation}
\paragraph{Deterministic bound on $Q_{ij}$
(conditioned on $\mathcal{E}^c$).}
Define $Q:=\frac{1}{n}\hat{\bX}_2\hat{\bX}_1^\top-\frac{1}{n}\bX_2\bX_1^\top$. On $\mathcal{E}^c$, expand $\hat{\bX}_k=\bX_k+(\hat{\bX}_k-\bX_k)$ to obtain
\begin{align*}
    Q &= \frac{1}{n}\big[\bX_2+(\hat{\bX}_2-\bX_2)\big]
    \big[\bX_1+(\hat{\bX}_1-\bX_1)\big]^\top
    -\frac{1}{n}\bX_2\bX_1^\top\\
    &= \underbrace{\frac{1}{n}\bX_2
    (\hat{\bX}_1-\bX_1)^\top}_{Q_1}
    +\underbrace{\frac{1}{n}(\hat{\bX}_2-\bX_2)
    \bX_1^\top}_{Q_2}
    +\underbrace{\frac{1}{n}(\hat{\bX}_2-\bX_2)
    (\hat{\bX}_1-\bX_1)^\top}_{Q_3}.
\end{align*}
By the triangle inequality and submultiplicativity of the operator norm,
\begin{align*}
    \|Q\|_{\mathrm{op}}
    &\leq\|Q_1\|_{\mathrm{op}}
    +\|Q_2\|_{\mathrm{op}}
    +\|Q_3\|_{\mathrm{op}}\\
    &\leq\frac{\|\bX_2\|_{\mathrm{op}}\,
    \|\hat{\bX}_1-\bX_1\|_{\mathrm{op}}}{n}
    +\frac{\|\hat{\bX}_2-\bX_2\|_{\mathrm{op}}\,
    \|\bX_1\|_{\mathrm{op}}}{n}
    +\frac{\|\hat{\bX}_2-\bX_2\|_{\mathrm{op}}\,
    \|\hat{\bX}_1-\bX_1\|_{\mathrm{op}}}{n}.
\end{align*}
On $\mathcal{E}^c$, Eq.~\ref{eq:omega_data} and Eqs.~\ref{eq:achievopbound2}--\ref{eq:achievopbound3}, together with $B_k^{(\mathrm{cov})}/(nd_k)\ge\beta$, give
\begin{align}
    \|Q_1\|_{\mathrm{op}}
    &\leq\frac{108\sigma^2
    \sqrt{(d_1+n)(d_2+n)}}{n}\cdot 2^{-\beta},
    \label{eq:achievq1}\\
    \|Q_2\|_{\mathrm{op}}
    &\leq\frac{108\sigma^2
    \sqrt{(d_1+n)(d_2+n)}}{n}\cdot 2^{-\beta},
    \label{eq:achievq2}\\
    \|Q_3\|_{\mathrm{op}}
    &\leq\frac{324\sigma^2
    \sqrt{(d_1+n)(d_2+n)}}{n}\cdot 2^{-2\beta}.
    \label{eq:achievq3}
\end{align}
By~\eqref{eq:sachisubsetbound2},
$n\ge C_2d\ge d$ for $C_2\ge 1$. Therefore, 
% $n\geq C_2 s\log(d_1d_2)\sigma^4 d/\varepsilon^2$. The low-distortion assumption $\varepsilon<\sigma^2\sqrt{C_2 s\log(d_1d_2)}$ gives
% $$
%     \frac{C_2 s\log(d_1d_2)\sigma^4 d}{\varepsilon^2}
%     > d
%     \quad\iff\quad
%     \varepsilon<\sigma^2\sqrt{C_2 s\log(d_1d_2)},
% $$
% hence $n>d$.  
% Then,
$$
    \frac{d_k+n}{n}=1+\frac{d_k}{n}
    \leq 1+\frac{d}{n}\le2
    \implies
    \frac{\sqrt{(d_1+n)(d_2+n)}}{n}\le2.
$$
% Finally, by the low-distortion assumption
% $\varepsilon<C_3\sigma^2$, we have
% $\log(C_3\sigma^2/\varepsilon)>0$, hence $\beta>0$
% and $2^{-2\beta}\leq 2^{-\beta}$.  Combining the above, we have
Also, $\beta\ge 2$ implies $2^{-2\beta}\le 2^{-\beta}$. Thus,
\begin{equation}\label{eq:achievq}
    |Q_{ij}|=|e_i^\top Q\,e_j|
    \leq\|Q\|_{\mathrm{op}}
    \leq(2\times 216 + 648)\sigma^2\cdot 2^{-\beta}
    = 1080\sigma^2\cdot 2^{-\beta}
    =:Q_{\max}.
\end{equation}

\paragraph{Probabilistic bound on $S_{ij}$.}
Since the observations have mean zero, $\mathbb{E}\big[[X^{(l)}_2]_i[X_1^{(l)}]_j\big]=[C_{21}]_{ij}$. Thus $S_{ij}$ is the average of $n$ independent, centered random variables. Each $[X_k^{(l)}]_i$ is mean-zero and $\sigma$-sub-Gaussian (as a coordinate of a $\sigma$-sub-Gaussian vector), so $[X_2^{(l)}]_i[X_1^{(l)}]_j$ is sub-exponential by Lemma~\ref{lem:subgaussianproduct}. By the triangle inequality for the Orlicz norm, $[X_2^{(l)}]_i[X_1^{(l)}]_j-[C_{21}]_{ij}$ has $\psi_1$ norm at most $K:=C'\sigma^2$ for a universal constant $C'>0$.

\noindent Applying Bernstein's inequality
(Theorem~\ref{thm:bernstein}) with $t=n\lambda,$ where $\lambda:=c_0\sigma^2\sqrt{\log(d_1d_2)/n}$ and $K=C'\sigma^2$,
\begin{equation}\label{eq:achievbernstein}
    \mathbb{P}\big(|S_{ij}|\geq\lambda\big)
    \leq 2\exp\!\left[-c\min\!\left(
    \frac{n\lambda^2}{K^2},\;
    \frac{n\lambda}{K}\right)\right].
\end{equation}
Since $n\ge C_2d\ge C_2\log(d_1d_2)$ (using $\log_2(d_1d_2)\le d_1+d_2=d$), taking $C_2\geq c_0^2/C'^2$ gives $\lambda/K\leq 1$, so the minimum is the first term and we have
\begin{equation}\label{eq:achievbernstein2}
    \mathbb{P}\big(|S_{ij}|\geq\lambda\big)
    \leq 2\exp\!\left(
    -\frac{cc_0^2}{C'^2}\log(d_1d_2)\right)
    =\frac{2}{(d_1d_2)^{(cc_0^2/C'^2)/\ln 2}}\le\frac{2}{(d_1d_2)^{cc_0^2/C'^2}}.
\end{equation}
Choosing $c_0\geq C'\sqrt{3/c}$ gives
$cc_0^2/C'^2\geq 3$, hence
\begin{equation}\label{eq:achievbernstein3}
    \mathbb{P}\big(|S_{ij}|\geq\lambda\big)
    \leq\frac{2}{(d_1d_2)^3}.
\end{equation}
Since $\mathbb{P}(\mathcal{E}^c)\geq 0.9999$,
conditioning on $\mathcal{E}^c$ gives
\begin{equation}\label{eq:achievbernstein4}
    \mathbb{P}\big(|S_{ij}|\geq\lambda
    \mid\mathcal{E}^c\big)
    \leq\frac{3}{(d_1d_2)^3}.
\end{equation}
\paragraph{Thresholding dominance.} Since $|Q_{ij}|\leq Q_{\max}$ deterministically on $\mathcal{E}^c$ and $\lambda^*=\lambda+Q_{\max}$, the inverse triangle inequality gives
\begin{equation}\label{eq:achievqsthres}
    \mathbb{I}\big[|S_{ij}+Q_{ij}|\geq\lambda^*\big]
    \leq\mathbb{I}\big[|S_{ij}|\geq\lambda\big].
\end{equation}
\paragraph{Setting up the post-threshold error analysis.} 
Let $\mathcal{S}:=\{(i,j)\in[d_2]\times[d_1]:
[C_{21}]_{ij}\neq 0\}$ denote the support of $C_{21}$,
with $|\mathcal{S}|\le s$.  Decompose the squared Frobenius
error as
\begin{equation}\label{eq:achievcrosscovdecomp}
    \mathbb{E}\big[\|\hat{C}_{21}^{\text{threshold}}
    -C_{21}\|_F^2\big|\mathcal{E}^c\big]
    =\underbrace{\sum_{(i,j)\in\mathcal{S}}
    \mathbb{E}\big[\big([\hat{C}_{21}^{\text{threshold}}]_{ij}
    -[C_{21}]_{ij}\big)^2\big|\mathcal{E}^c\big]}_{
    |\mathcal{S}|\text{ non-zero entries}}
    +\underbrace{\sum_{(i,j)\notin\mathcal{S}}
    \mathbb{E}\big[[\hat{C}_{21}^{\text{threshold}}]_{ij}^2
    \big|\mathcal{E}^c\big]}_{
    d_1d_2-|\mathcal{S}|\text{ zero entries}}.
\end{equation}

\paragraph{Non-zero entries: $(i,j)\in\mathcal{S}$.}
Write $\gamma_{ij}:=[C_{21}]_{ij}\neq 0$.  The
thresholded estimate satisfies
$[\hat{C}_{21}^{\text{threshold}}]_{ij}\in
\{[\hat{C}_{21}]_{ij},\,0\}$, so
\begin{align}
    \mathbb{E}\big[\big([\hat{C}_{21}^{\text{threshold}}]_{ij}-\gamma_{ij}\big)^2\big|\mathcal{E}^c\big]=\mathbb{E}\big[(S_{ij}+Q_{ij})^2\,\mathbb{I}[|[\hat{C}_{21}]_{ij}|\geq\lambda^*]\big|\mathcal{E}^c\big]+\gamma_{ij}^2\,\mathbb{P}\big(|[\hat{C}_{21}]_{ij}|<\lambda^*\big|\mathcal{E}^c\big).
    \label{eq:achievnonzerodecomp}
\end{align}
% \noindent\emph{First term in Eq.~\ref{eq:achievnonzerodecomp}.} Since $0\le\mathbb{I}[|[\hat{C}_{21}]_{ij}|\ge\lambda^*|\mathcal{E}^c]\le1$, the first
% term satisfies
% \begin{equation}\label{eq:achievnonzerofirst}
%     \mathbb{E}\big[(S_{ij}+Q_{ij})^2
%     \big|\mathcal{E}^c\big]
%     \leq 2\,\mathbb{E}\big[S_{ij}^2
%     \big|\mathcal{E}^c\big]+2Q_{\max}^2\le\frac{2\mathbb{E}[S^2_{ij}]}{\mathbb{P}(\mathcal{E}^c)}+2Q^2_{\max}
%     \leq\frac{C''\sigma^4}{n}+2Q_{\max}^2,
% \end{equation}
% where $\mathbb{E}[S_{ij}^2]=\mathrm{Var}(S_{ij})\leq C''\sigma^4/(2n)$ by Proposition~2.8.1 of~\cite{vershynin2018high} and $|\gamma_{ij}|\leq\|C\|_{\mathrm{op}}\leq\sigma^2$.\\
\noindent\emph{First term in Eq.~\ref{eq:achievnonzerodecomp}.}
Since
$0\le\mathbb{I}[|[\hat{C}_{21}]_{ij}|\ge\lambda^*]\le1$,
the first term is bounded by
$\mathbb{E}[(S_{ij}+Q_{ij})^2\mid\mathcal{E}^c]$.

Recall that $S_{ij}$ is the average of $n$ independent, centered
sub-exponential random variables, each with $\psi_1$ norm at most
$C'\sigma^2$. By Lemma~\ref{lem:subexpmoment} and independence across observations,
we may choose a sufficiently large universal constant $C''>0$ such that
\[
    \mathbb{E}[S_{ij}^2]
    =\mathrm{Var}(S_{ij})
    \le \frac{C''\sigma^4}{4n}.
\]
Using $|Q_{ij}|\le Q_{\max}$ on $\mathcal{E}^c$ and
$\mathbb{P}(\mathcal{E}^c)\ge0.9999\ge1/2$, we obtain
\begin{equation}\label{eq:achievnonzerofirst}
\begin{aligned}
    \mathbb{E}\big[(S_{ij}+Q_{ij})^2
        \mid\mathcal{E}^c\big]
    \le 2\,\mathbb{E}\big[S_{ij}^2
        \mid\mathcal{E}^c\big]+2Q_{\max}^2
    \le \frac{2\mathbb{E}[S_{ij}^2]}
        {\mathbb{P}(\mathcal{E}^c)}+2Q_{\max}^2&\le \frac{C''\sigma^4}{n}+2Q_{\max}^2.
\end{aligned}
\end{equation}
\noindent\emph{Second term in Eq.~\ref{eq:achievnonzerodecomp}.} For the second term, we bound $\gamma_{ij}^2\,\mathbb{I}[\text{zeroed}]$ deterministically. By the triangle inequality, on the event $|[\hat{C}_{21}]_{ij}|<\lambda^*$,
\[
    |\gamma_{ij}|\leq|\gamma_{ij}+S_{ij}+Q_{ij}|
    +|S_{ij}+Q_{ij}|<\lambda^*+|S_{ij}+Q_{ij}|.
\]
Squaring both sides and using $(a+b)^2\leq 2a^2+2b^2$,
% \begin{equation}\label{eq:achievnonzerogammabound}
%     \gamma_{ij}^2\,\mathbb{P}\big[|[\hat{C}_{21}]_{ij}|<\lambda^*\big|\mathcal{E}^c\big]\leq 2(\lambda^*)^2+2(S_{ij}+Q_{ij})^2.
% \end{equation}
\begin{equation}\label{eq:achievnonzerogammabound}
    \gamma_{ij}^2\,\mathbb{I}\big[|[\hat{C}_{21}]_{ij}|<\lambda^*\big]
    \leq 2(\lambda^*)^2+2(S_{ij}+Q_{ij})^2.
\end{equation}
Taking conditional expectations,
\begin{align}
    \mathbb{E}\big[\gamma_{ij}^2\mathbb{I}[|[\hat{C}_{21}]_{ij}|<\lambda^*]
    \big|\mathcal{E}^c\big]
    &\leq 2(\lambda^*)^2
    +2\,\mathbb{E}\big[(S_{ij}+Q_{ij})^2
    \big|\mathcal{E}^c\big]\nonumber\\
    &\leq 2(2\lambda^2+2Q_{\max}^2)
    +2\!\left(\frac{C''\sigma^4}{n}
    +2Q_{\max}^2\right)\nonumber\\
    &=\frac{4c_0^2\sigma^4\log(d_1d_2)}{n}
    +\frac{2C''\sigma^4}{n}+8Q_{\max}^2,
    \label{eq:achievnonzerosecond}
\end{align}
where we used $(\lambda^*)^2=(\lambda+Q_{\max})^2
\leq 2\lambda^2+2Q_{\max}^2$ and
Eq.~\ref{eq:achievnonzerofirst}.\\
\noindent\emph{Non-zero entries bound.} Adding~\eqref{eq:achievnonzerofirst}
and~\eqref{eq:achievnonzerosecond},
\begin{align}
    \mathbb{E}\big[\big(
    [\hat{C}_{21}^{\text{threshold}}]_{ij}
    -\gamma_{ij}\big)^2\big|\mathcal{E}^c\big]
    &\leq\frac{C''\sigma^4}{n}+2Q_{\max}^2
    +\frac{4c_0^2\sigma^4\log(d_1d_2)}{n}
    +\frac{2C''\sigma^4}{n}+8Q_{\max}^2\nonumber\\
    &\leq\frac{C'''\sigma^4\log(d_1d_2)}{n}
    +10Q_{\max}^2,
    \label{eq:achievnonzerofinal}
\end{align}
where $C''':=3C''+4c_0^2$ and we used
$1/n\leq\log(d_1d_2)/n$.
Summing over at most $s$ non-zero entries:
\begin{equation}\label{eq:achievnonzerototal}
    \sum_{(i,j)\in\mathcal{S}}
    \mathbb{E}\big[\big(
    [\hat{C}_{21}^{\text{threshold}}]_{ij}
    -\gamma_{ij}\big)^2\big|\mathcal{E}^c\big]
    \leq\frac{C'''\sigma^4 s\log(d_1d_2)}{n}
    +10sQ_{\max}^2.
\end{equation}
\vspace{3mm}
\paragraph{Zero entries: $(i,j)\notin\mathcal{S}$.}
Since $[C_{21}]_{ij}=0$, the thresholded estimate
satisfies
$[\hat{C}_{21}^{\text{threshold}}]_{ij}
=(S_{ij}+Q_{ij})\,\mathbb{I}[|S_{ij}+Q_{ij}|
\geq\lambda^*]$.  By~\eqref{eq:achievqsthres},
\begin{equation}\label{eq:achievzerodecomp}
    \mathbb{E}\big[
    [\hat{C}_{21}^{\text{threshold}}]_{ij}^2
    \big|\mathcal{E}^c\big]
    \leq 2\,\mathbb{E}\big[S_{ij}^2\,
    \mathbb{I}[|S_{ij}|\geq\lambda]
    \big|\mathcal{E}^c\big]
    +2Q_{\max}^2\,\mathbb{P}\big[|S_{ij}|
    \geq\lambda\big|\mathcal{E}^c\big].
\end{equation}
\noindent\emph{Second term in Eq.~\ref{eq:achievzerodecomp}.} By~\eqref{eq:achievbernstein4}, the second term satisfies
\begin{equation}~\label{eq:achievzerosecond}
2Q^2_{\max}\mathbb{P}\big(|S_{ij}|\geq\lambda\big|\mathcal{E}^c\big)\leq\frac{6Q^2_{\max}}{(d_1d_2)^3}.
\end{equation}
\noindent\emph{First term in Eq.~\ref{eq:achievzerodecomp}.}
Since $\mathbb{P}(\mathcal{E}^c)\geq0.9999$,
\begin{equation}~\label{eq:achievzerofirst1}
    \mathbb{E}\big[S^2_{ij}\mathbb{I}[|S_{ij}|\geq\lambda]\big|\mathcal{E}^c\big]\leq2\mathbb{E}\big[S^2_{ij}\mathbb{I}[|S_{ij}|\geq\lambda]\big].
\end{equation}
\noindent By the tail-sum formula for expectations,
\begin{align}
    \mathbb{E}\big[S^2_{ij}\mathbb{I}[|S_{ij}|\geq\lambda]\big] &=\int_0^\infty \mathbb{P}\big[S^2_{ij}\mathbb{I}[|S_{ij}|\geq\lambda]\geq u\big]du\nonumber\\
    &=\int_0^\infty \mathbb{P}\big[|S_{ij}|\cdot\mathbb{I}[|S_{ij}|\geq\lambda]\geq t\big]\cdot2t dt\nonumber\\
    &=\underbrace{\int_0^\lambda \mathbb{P}\big[|S_{ij}|\geq\lambda\big]\cdot2tdt}_{t<\lambda}+\underbrace{\int_\lambda^\infty \mathbb{P}\big[|S_{ij}|\geq t\big]\cdot 2tdt}_{t\geq\lambda}\nonumber\\
    &=\lambda^2\mathbb{P}\big[|S_{ij}|\geq\lambda\big]+\int_\lambda^\infty \mathbb{P}\big[|S_{ij}|\geq t\big]\cdot 2tdt\nonumber\\
    &\le\frac{3c_0^2\sigma^4\log(d_1d_2)}{n(d_1d_2)^3}+\int_\lambda^K 2\exp\big(-c\frac{nt^2}{C'^2\sigma^4}\big)\cdot 2tdt+\int_K^\infty 2\exp\big(-c\frac{nt}{C'\sigma^2}\big)\cdot 2tdt\nonumber\\
    &\le\frac{3c_0^2\sigma^4\log(d_1d_2)}{n(d_1d_2)^3}+\frac{2C'^2\sigma^4}{cn(d_1d_2)^3}+\frac{4C'^2\sigma^4(cn+1)}{c^2n^2}e^{-cn}\nonumber\\
    &\leq\frac{3c_0^2\sigma^4\log(d_1d_2)}{n(d_1d_2)^3}+\frac{3C'^2\sigma^4}{cn(d_1d_2)^3}~\label{eq:achievzerofirst2}
\end{align}
where the second line follows by the substitution $u=t^2$, the fourth line by evaluating the integral over $[0,\lambda]$, the fifth line by Eq.~\ref{eq:achievbernstein} and Eq.~\ref{eq:achievbernstein3}, the sixth line by evaluating the integrals and using $cc_0^2/C'^2\ge3$, and the last line by Eq.~\ref{eq:sachisubsetbound2} and $C_2\ge32/c$, which imply $cn\ge32\log(d_1d_2)$.\\
\noindent \emph{Zero entries bound.} Combining Eqs.~\ref{eq:achievzerodecomp}--\ref{eq:achievzerofirst2},
\begin{equation}~\label{eq:achievzerofinal1}
    \mathbb{E}[[\hat{C}^{\text{threshold}}_{21}]^2_{ij}|\mathcal{E}^c]\le\frac{12c_0^2\sigma^4\log(d_1d_2)}{n(d_1d_2)^3}+\frac{12C'^2\sigma^4}{cn(d_1d_2)^3}+\frac{6Q^2_{\max}}{(d_1d_2)^3}\le\frac{C''''\sigma^4\log(d_1d_2)}{n(d_1d_2)^3}+\frac{6Q^2_{\max}}{(d_1d_2)^3},
\end{equation}
where the last step follows from $\log(d_1d_2)\geq1$, with $C''''=12c_0^2+\frac{12C'^2}{c}$.\\
Hence,
\begin{equation}~\label{eq:achievzerofinal2}
\begin{aligned}
    \sum_{(i,j)\notin\mathcal{S}}
\mathbb{E}[[\hat{C}^{\text{threshold}}_{21}]^2_{ij}
|\mathcal{E}^c]
&\leq(d_1d_2-|\mathcal{S}|)\left(
\frac{C''''\sigma^4\log(d_1d_2)}{n(d_1d_2)^3}
+\frac{6Q^2_{\max}}{(d_1d_2)^3}\right)\\
&\leq\frac{C''''\sigma^4\log(d_1d_2)}{n(d_1d_2)^2}
+\frac{6Q^2_{\max}}{(d_1d_2)^2}
\end{aligned}
\end{equation}
\paragraph{Final bound on $\mathbb{E}[\|\hat{C}_{21}^{\text{threshold}}-C_{21}\|^2_F|\mathcal{E}^c]$.}
Adding~\eqref{eq:achievnonzerototal}
and~\eqref{eq:achievzerofinal2},
\begin{align}
    \mathbb{E}\big[\|\hat{C}_{21}^{\text{threshold}}
    -C_{21}\|_F^2\big|\mathcal{E}^c\big]
    &\leq\frac{C'''\sigma^4 s\log(d_1d_2)}{n}
    +10sQ_{\max}^2
    +\frac{C''''\sigma^4\log(d_1d_2)}{n(d_1d_2)^2}
    +\frac{6Q_{\max}^2}{(d_1d_2)^2}\nonumber\\
    &\leq\frac{C_4\sigma^4 s\log(d_1d_2)}{n}
    +16sQ_{\max}^2,
    \label{eq:achievfinal}
\end{align}
where the last step uses $s\geq 1$ and
$1/(d_1d_2)^2\leq 1\leq s$ to absorb the
$(d_1d_2)^{-2}$ terms, and $C_4:=C'''+C''''$.\\
\noindent By Jensen's inequality
($\mathbb{E}[\|A\|_F]\leq\sqrt{\mathbb{E}[\|A\|_F^2]}$)
and $\sqrt{a+b}\leq\sqrt{a}+\sqrt{b}$,
\begin{equation}\label{eq:achievcrossfinal}
    \mathbb{E}\big[\|\hat{C}_{21}^{\text{threshold}}
    -C_{21}\|_F\big|\mathcal{E}^c\big]
    \leq\sqrt{\frac{C_4\sigma^4 s\log(d_1d_2)}{n}}
    +4\sqrt{s}\,Q_{\max}.
\end{equation}
Therefore,
\begin{equation}\label{eq:achievcrossfinal2}
    \sqrt{2}\,\mathbb{E}\big[
    \|\hat{C}_{21}^{\text{threshold}}
    -C_{21}\|_F\big|\mathcal{E}^c\big]
    \leq\sqrt{\frac{2C_4\sigma^4 s\log(d_1d_2)}{n}}
    +4\sqrt{2s}\,Q_{\max}.
\end{equation}
% \noindent Recall $Q_{\max}=1080\sigma^2\cdot 2^{-\beta}$ and
% $\beta=2\log(C_3\sigma^2/\varepsilon)$, so
% $2^{-\beta}=(\varepsilon/(C_3\sigma^2))^2$. Also, $\varepsilon<C_3\sigma^2$ under the reasonable distortion assumption. Then
% $$
%     4\sqrt{2s}\,Q_{\max}
%     =\frac{4320\sqrt{2s}\,\varepsilon^2}
%     {C_3^2\sigma^2}<\frac{4320\sqrt{2s}\varepsilon}{C_3}.
% $$
By Eq.~\ref{eq:sachibeta}, we have
\[
2^{-\beta}\le\frac{\varepsilon}{C_3\sigma^2\sqrt{s}}.
\]
Thus, with $C_3\ge 34560\sqrt{2}$, 
\[ 4\sqrt{2s}Q_{\max}=4320\sqrt{2s}\sigma^2 \cdot 2^{-\beta}\le\frac{4320\sqrt{2}}{C_3}\varepsilon\le\frac{\varepsilon}{8}.\]

\noindent By~\eqref{eq:sachisubsetbound2},
$n\ge\frac{C_2\sigma^4 s\log(d_1d_2)}{\varepsilon^2}$. Substituting into the first term
of~\eqref{eq:achievcrossfinal2},
\[
    \sqrt{\frac{2C_4\sigma^4 s\log(d_1d_2)}{n}}
    \leq\varepsilon\sqrt{\frac{2C_4}{C_2}}\le\frac{\varepsilon}{8},\qquad\text{if }C_2\ge 128 C_4.
\]
Hence,
\begin{equation}\label{eq:achievcrossfinal3}
    \sqrt{2}\,\mathbb{E}\big[
    \|\hat{C}_{21}^{\text{threshold}}
    -C_{21}\|_F\big|\mathcal{E}^c\big]
    \leq\frac{\varepsilon}{8}+\frac{\varepsilon}{8}
    =\frac{\varepsilon}{4}.
\end{equation}
\paragraph{Choice of universal constants.}
First fix the universal constants $c,C'$ in Eq.~\ref{eq:achievbernstein}. Choose $c_0$ such that $cc_0^2/C'^2\ge3$. Choose a sufficiently large universal constant $C''>0$ such that
\[
    \mathbb{E}[S_{ij}^2]\le\frac{C''\sigma^4}{4n},
\]
as established in the argument for Eq.~\ref{eq:achievnonzerofirst}, and set
\[
    C'''=3C''+4c_0^2,\qquad
    C''''=12c_0^2+\frac{12C'^2}{c},\qquad
    C_4=C'''+C''''.
\]
Finally choose
\[
    C_2\ge\max\left\{
        10000,\frac{c_0^2}{C'^2},\frac{32}{c},128C_4
    \right\},
    \qquad
    C_3\ge34560\sqrt{2}.
\]
\paragraph{Total error bound.} Collecting all bounds,
\begin{align*}
    \mathbb{E}[\|\hat{C}-C\|_F]
    &\leq\underbrace{\frac{\varepsilon}{10000}}_{\text{bad
    event}}
    +\underbrace{\frac{\varepsilon}{7}
    +\frac{\varepsilon}{7}}_{\text{self-cov},\;k=1,2}
    +\underbrace{\frac{\varepsilon}{4}}_{\text{cross-cov}}\\
    &=\frac{\varepsilon}{10000}+\frac{2\varepsilon}{7}
    +\frac{\varepsilon}{4}
    =\varepsilon\left(\frac{1}{10000}
    +\frac{2}{7}+\frac{1}{4}\right)
    <\varepsilon,
\end{align*}
since $\frac{1}{10000}+\frac{2}{7}+\frac{1}{4}
=\frac{1}{10000}+\frac{15}{28}<1$.
This completes the proof of
Theorem~\ref{thm:sachievability}.

\end{document}